\documentclass[lettersize,draftcls,onecolumn]{IEEEtran}
\usepackage{amsmath,amsfonts}
\usepackage{array}
\usepackage[caption=false,font=normalsize,labelfont=sf,textfont=sf]{subfig}
\usepackage{textcomp}
\usepackage{stfloats}
\usepackage{url}
\usepackage{verbatim}
\usepackage{graphicx}
\usepackage{cite}
\usepackage{color}
\usepackage{algorithm}
\usepackage{algpseudocode} 
\usepackage{hyperref}

\usepackage{amsmath,amssymb,amsthm}
\theoremstyle{plain}
\newtheorem{theorem}{Theorem}[section]
\newtheorem{lemma}{Lemma}[section]
\newtheorem{proposition}{Proposition}[section]
\newtheorem{corollary}{Corollary}[section]

\theoremstyle{remark}

\newtheorem{remark}{Remark}

\newcommand{\Pro}{\mathrm{P}}

\newcommand{\Exp}{\mathrm{E}}

\newcommand{\cS}{S}

\newcommand{\cF}{\mathcal{F}}
\newcommand{\bN}{\mathbb{N}}

\newcommand{\cJ}{\mathcal{J}}

\newcommand{\cV}{\mathcal{V}}

\newcommand{\ba}{\boldsymbol{a}}

\DeclareMathOperator*{\argmax}{arg\,max\;}

\begin{document}

\title{Active Sequential Signal Detection \\ with Asynchronous Decisions}

\author{Yiming Xing,~\IEEEmembership{Member,~IEEE,} and
Georgios Fellouris,~\IEEEmembership{Member,~IEEE}
\thanks{Yiming Xing (email: yimingx4@tongji.edu.cn) is with the School of Mathematical Sciences, Tongji University, Shanghai, China.
Georgios Fellouris (email: fellouri@illinois.edu) is with the Department of Statistics, University of Illinois at Urbana-Champaign, Urbana-Champaign, IL, USA.}
\thanks{Part of this work was presented at 2026 IEEE International Symposium on Information Theory \cite{ISIT2026_GF}.}
\thanks{This research was supported by ...}
\thanks{Manuscript received ...; revised ...}}


\markboth{Journal of \LaTeX\ Class Files,~Vol.~1, No.~2, December~2023}%
{Shell \MakeLowercase{\textit{et al.}}: A Sample Article Using IEEEtran.cls for IEEE Journals}

\IEEEpubid{0000--0000~\copyright~2023 IEEE}

\maketitle

\begin{abstract}
The problem of quickly detecting signals among multiple sequentially observed data streams, 
subject to constraints on the familywise false-alarm and missed-detection probabilities, 
is considered under an active sampling setup where only one stream can be sampled at each time instant. 
Under this setup, it is in general impossible to minimize, even asymptotically, 
the expected decision time simultaneously for every stream. 
We therefore consider minimizing the expected time until either the $k_{th}$ signal declaration or termination simultaneously for every $k$. 
We first establish non-asymptotic lower bounds for these criteria. 
For $k>|B|$, where $B$ denotes the true subset of signals, we show that these lower bounds are attained asymptotically as the error levels vanish by any policy that stops sampling and makes a decision for each stream according to a local Sequential Probability Ratio Test. 
For $1\leq k\leq |B|$, however, the performance depends critically on the sampling rule, and we propose one with a simple exploration mechanism that prevents switching away from signals prematurely because of early unfavorable observations. We show that, under a mild relative-rate condition on the two error levels, the resulting policy attains all the lower bounds simultaneously. Numerical studies are presented to compare the proposed policy with an oracle policy and a follow-the-leader policy.
\end{abstract}


\section{Introduction}
Consider a detection system comprising multiple detectors, each monitoring  an environment. It is of interest to identify, based on observations collected in real-time, whether there is a signal in each of these environments as quickly as possible while guaranteeing certain global reliability constraints.
Such a problem arises in many scientific and engineering fields, where the ``environments" may be any scenarios or media, and the ``signals" may be any phenomena or segments with certain characteristics of interest, e.g., detecting intrusions in radar arrays \cite{Ref4MissleDet_2004}, anomalies or outliers in surveillance systems \cite{Ref4AnomalyDet_2009}, influenced endpoints in clinical trials \cite{Bartroff_Book_Clinicaltrials}, unoccupied channels in communication networks \cite{Ref4SpectrumSensing_2016}, unexpected or useful patterns in databases \cite{Ref4PatternMining_2017}, matching pairs in gene-association studies \cite{Ref4GeneAsso_2021}, frauds in financial markets or other platforms \cite{Ref4FraudDet_2022}, etc.
If the distinguishing characteristics of signals are specified as the alternative hypotheses and those of non-signals, i.e., noises, are specified as the null hypotheses, such a problem can be naturally formulated as a sequential multiple testing problem.

Early studies of this problem include \cite{Bartroff_2010,De_Baron_Seq_Bonf, Step_up_down,Bartroff_2014_FWER}, where sequential policies, i.e., policies where the number of observations collected in each stream is not predetermined but adaptively determined based on the collected observations, were proposed and were shown to control two types of error metrics below arbitrary, user-specified levels.
Besides, these sequential policies were shown, in numerical studies, to outperform their fixed-sample-size counterparts.

Later works go beyond controlling certain error metrics and consider theoretically minimizing the expectation of an objective function about the number of observations.
Such objective functions include:
(1) \textit{the first time at which a stream is detected as a signal, if any}, e.g., \cite{Dragalin1996Metrika, Lai2011quickest, Malloy_Nowak_2013, Fellouris2013unstructured, Ali2016quickest, Fellouris_noniid}, 
(2) \textit{a common time at which all streams are identified as either a signal or a noise}, e.g., \cite{Kobi_2015_active, Kobi_2018_heterogeneous, Kobi_2020_composite, Kobi_2022_switching, Kobi_2023_hiera,Song_prior, Song_AoS, Aris_IEEE, Aris_TIT2025, chaudhuri2024joint, SPL2026, ISIT2026_CJ_CS, SS2026+, Xing_config_2026},
and (3) \textit{a function of the times at which streams are identified}, e.g., \cite{Malloy_Nowak_2014, Kobi_2014_index, Kobi_2015_AO, Kobi_2019_nonlinear, PaperII, ITW2024}.

The work \cite{PaperIII} is the first and, as far as the authors know, currently the only one where 
more than one objective functions are simultaneously minimized in an asymptotic sense.  
However, this  is achieved in a \emph{full sampling} setup, where all streams can be observed at every time instant, even after some of them have reached a decision.
Our focus in this work is on the \emph{active sampling} setup,
where only a subset of  streams can be observed at each time instant.
This constraint may arise from practical limitations such as sampling budget, communication bandwidth and computing capacity, and has been considered, e.g., in 
\cite{Chernoff1959, Naghshvar2013, ControlledSensing, ControlledSensing_NonUniformCost, Kobi_2015_active, Kobi_2020_composite, ControlledSensing_Composite, Aris_IEEE, Aris_TIT2025, ISIT2026_CJ_CS}
as well as in \cite{Qunzhi_2021TIT, Qunzhi2023, Fellouris_2024_CDwithCS, Ana_2024TIT} that study a closely-related problem of sequential detection of changepoints. 
The solution to this active sequential multiple testing problem requires the specification of a sampling rule that determines which streams to observe at each time instant, 
a decision rule that determines which streams to declare to be signals, and 
a termination time that determines when to stop the policy with the declaration that all signals have already been detected. 





Specifically,  we consider $K> 1$ independent data streams and postulate two simple hypotheses for each of them,
we require controlling  the probabilities of falsely detecting any noise and missing any signal below  given levels $\alpha$ and $\beta$, respectively,  
and we assume that only one stream can be observed at each time instant.
For this problem, it has been  shown in \cite{Aris_IEEE} that the expected time until termination, 
which is also the expected total sample size, 
is asymptotically minimized  when the  local log-likelihood ratio (LLR) statistic of all streams are simultaneously outside an interval of the form $(-b,a)$, where $a,b>0$, 
as in Wald's Sequential Probability Ratio Test (SPRT) \cite{Wald_Book}, the streams whose LLRs upon stopping are declared as signals, and 
 a specifically designed sampling rule is adopted.
As we show in this work, no particular sampling rule is needed for this purpose. 
The same asymptotic optimality property holds for any sampling rule as long as it does not sample a stream after its LLR has exited $(-b,a)$. Therefore, we need to introduce additional criteria in order to select in a principled way the sampling rule.

In contrast to the full sampling setup in \cite{PaperIII}, in the active sampling setup it is impossible to minimize, even in an asymptotic sense, the 
expected decision time simultaneously for every stream.
Indeed, if the goal is to identify a specific stream as quickly as possible, then clearly we should prioritize collecting observations from that stream, thereby delaying the identification of other streams. 
Motivated by this observation, as well as the fact that it  often more critical to quickly identify signals rather than identify noises, in this work we propose to 
minimize the expectations of \textit{the order statistics of the times at which signals are  declared}. 
More specifically, we propose to minimize the expected time until \textit{either declaring the detection of the $k_{th}$ signal or declaring that there are fewer than $k$ signals and terminating the policy} simultaneously for every $1\leq k\leq K$. Note that this criterion is well-defined even if fewer than $k$ signals are declared. 

For the proposed criteria, we first establish
intuitive, universal, non-asymptotic lower bounds for them (Section \ref{section: universal lower bound}).
These lower bounds point to an oracle sampling rule, according to which the streams are ordered as signals first, from the easiest to the most difficult, and noises next, in an arbitrary order, and then an SPRT is applied to each of them in this order.
The difficulty level of a signal is quantified by the Kullback-Leibler (KL) divergence from its signal distribution to its noise distribution. 

However, since the true subset of signals is unknown, this oracle sampling rule is not directly applicable, and we need to find a sampling rule that ``learns'' the true subset of signals. A sampling rule that has been proposed in the literature \cite[Section IV.B]{Kobi_2015_active} is to  ``follow-the-leader", that is, to always sample the stream with the largest LLR. 
This sampling rule is efficient for detecting the last signal \cite[Theorem 4]{Kobi_2015_active}, but not necessarily for detecting easier signals earlier. 
Indeed, if the LLR of the easiest signal
happens to go downward based on the first few observations, which is an event of positive probability,
the ``follow-the-leader'' policy will set it aside and switch to other streams that will on average take longer to be detected even if they are signals.

Motivated by this observation, we  incorporate an exploration mechanism that helps us to focus on signals with precision.  
Specifically, we order all streams according to their difficulty levels as if they are indeed signals, 
and we start by sampling the first stream until its LLR is either above the detection threshold $a>0$, 
or below a negative value $-b'$, which is less extreme than the futility threshold $-b$.
We repeat this process with the second stream, etc., until all streams have been traversed.  
We refer to this as Phase I of our proposed sampling rule.

Our main result of this work is that, irrespective of how we sample after Phase I, this sampling rule asymptotically minimizes the expected time until the $k_{th}$ signal declaration or termination simultaneously for every $1\leq k\leq K$, as long as the tuning parameter $b'$ is sufficiently large but significantly smaller than $b$.
Numerical studies show that the proposed policy can substantially improve the early detection of easy signals 
relative to the full-time follow-the-leader policy, 
while retaining similar performance for the detection of all signals and for termination.

A closely related   line of research  is in the good-arm identification problem in the multi-armed bandit literature. 
Specifically, \cite{GAI2019ML,GAI2025AISTAT} consider the problem of identifying arms whose means exceed a given threshold, which are referred to as good arms, 
while controlling the probability of misidentification. These works also consider multiple performance criteria associated with the successive identification of good arms. 
Specifically, if the true number of good arms is $m\geq 0$, they consider the expected time until $k$ good arms are identified, for each $1\leq k\leq m$, and the expected time until termination. 
Our formulation differs from these works in several important respects. 
First, we allow general binary hypotheses for each stream.
Second, we allow distinct error levels for the two types of errors.
Third, we minimize $\Exp_B[T_k\wedge T_{\text{stop}}]$ for every $1\leq k\leq K$, where $B$ represents the true subset of signals.
This criterion is well defined without knowledge of the unknown number $|B|$ of signals and, in particular, also accounts for the case $k>|B|$.
The main difference, however, is that  we establish a first-asymptotic optimality with respect to all these proposed objective functions,  in contrast to the rate-optimality  and minimax optimality results in 
results in \cite{GAI2019ML} and \cite{GAI2025AISTAT}, respectively. 
Related problems in this literature include identifying a single good arm \cite{Murphy_sampling,FirstGAI2026JMLR}, and classifying every arm as either good or bad (also known as thresholding bandit) \cite{TB2016ICML,TB2022AISTAT}.


The rest of the paper is organized as follows. In Section \ref{section: problem formulation}, we formulate the problem. In Section \ref{section: universal lower bound}, we establish universal lower bounds for the proposed performance criteria. In Section \ref{section: active SPRT}, we introduce the active SPRT and establish its error control and asymptotic optimality for the expected total sample size. 
In Section \ref{section: quick detection of signals}, we introduce the proposed policy and establish its asymptotic optimality for the expected successive signal declaration times.
In Section \ref{section: numerical studies}, we present numerical studies. In Section \ref{section: conclusion}, we conclude the paper and discuss directions for future research. Long proofs and supporting lemmas are deferred to the Appendix.

\section{Problem formulation} \label{section: problem formulation}
We consider $K$ independent data streams, each generating a sequence of iid random elements. 
For every  $i\in[K]:=\{1,\ldots,K\}$, we denote by $\{X_i(n): n \in \bN\}$ the sequence of random elements generated by stream $i$.
We assume that their common density with respect to some $\sigma$-finite measure, $\nu_i$, is either $f_i$ or $g_i$, and that the Kullback-Leibler (KL) divergences between $f_i$ and $g_i$ are positive and finite, i.e., 
\begin{equation} \label{only assumption}
\begin{aligned}
    I_i & := \int g_i \log \left( \frac{g_i}{f_i}\right) d\nu_i\in(0,\infty), \\
    J_i & := \int f_i \log \left(\frac{f_i}{g_i} \right) d\nu_i\in(0,\infty).
\end{aligned}
\end{equation}
We refer to stream $i$ as a \emph{noise} (resp. \emph{signal}) if its density is $f_i$ (resp. $g_i$),
and 
to $I_i$ as its \emph{signal KL-divergence}.
We denote by $\Pro_B$ the joint distribution of all streams if the true subset of signals is $B$, i.e., if $X_i(1)\sim g_i$ for every $i\in B$ and $X_i(1)\sim f_i$ for every $i\notin B$, and by $\Exp_B$ the corresponding expectation. 
Our goal is to detect each signal as quickly as possible and with prescribed accuracy.
To avoid possible confusion, we use \emph{signal / noise declaration} to refer to the action of identifying 
a stream 
as a signal / noise regardless of whether this declaration is correct, 
and use \emph{signal detection} to refer to the objective 
of correctly identifying a signal 
stream as a signal. 

Before diving into the multi-stream and active sampling formulation that we consider in this work, it is useful to briefly review the standard, single-stream formulation.

\subsection{The single-stream formulation}
Let us fix an arbitrary stream $i \in [K]$. We say that a type-I error (or false positive, or false alarm)  occurs in stream $i$ if it is declared as a signal when it is actually a noise.
Similarly, we say that  a type-II error (or false negative, or missed detection) occurs in stream $i$ if it is declared as a noise when it is actually a signal. In the traditional single-stream formulation of sequential testing,   the problem is to find a  finite random time, say $T_i$, which is a stopping time with respect to the filtration generated by stream $i$, $ \cF_i(n) := \sigma(X_i(1), \ldots, X_i(n))$ for $n \in \bN$, 
and an $\cF_i(T_i)$-measurable random variable, $D_i$, that takes values in $\{0,1\}$, with $D_i=1$ (resp. 0) indicating that stream $i$ is declared to be a signal (resp. noise), so that the probabilities of a type-I and a type-II error in stream $i$ are at most $\alpha$ and $\beta$, respectively, i.e.,
\begin{align} \label{error_local}  
    \Pro_{i,0}(T_i<\infty) = \Pro_{i,1}(T_i<\infty) = 1, \quad 
    \Pro_{i,0}(D_i=1)\leq\alpha, \quad \Pro_{i,1}(D=0) \leq \beta,
\end{align}  
where, only restricted to this subsection, we use $\Pro_{i,0}$ (resp. $\Pro_{i,1}$) to denote the distribution of stream $i$ if it is a noise (resp. signal), and use $\Exp_{i,0}$ (resp. $\Exp_{i,1}$) to denote the corresponding expectation. 
Denote by $\Delta_i(\alpha,\beta)$ all $(T_i,D_i)$ that satisfies \eqref{error_local}. 
Wald established the following lower bounds on the expected sample sizes when stream $i$ is a signal and when it is a noise: 
\begin{align*} 
\begin{split} 
\inf_{(T_i,D_i)\in\Delta_i(\alpha,\beta)} \Exp_{i,1}[ T_i ] &\geq \frac{d(1-\beta,\alpha)}{I_i} \\
\inf_{(T_i,D_i)\in\Delta_i(\alpha,\beta)} \Exp_{i,0}[ T_i ] &\geq \frac{ d(1-\alpha, \beta)}{J_i},
\end{split}
\end{align*}
where $d(x,y)$ represents the 
Kullback-Leibler (KL) divergence between two Bernoulli distributions with success probabilities $x$ and $y$, respectively, i.e., 
\begin{equation} \label{func d}
d(x,y):=x\log\left(\frac{x}{y}\right)+(1-x)\log\left(\frac{1-x}{1-y}\right), \quad x,y\in(0,1),
\end{equation}
(see, e.g., \cite{Wald_Book, Tartakovsky_Book}). 
Moreover, Wald shows that both these lower bounds are achieved, in a  certain sense, by 
the so-called Sequential Probability Ratio Test (SPRT) \cite{Wald_Book, Tartakovsky_Book}. To be specific, denote by 
$\lambda_i(n)$
the log-likelihood ratio statistic (LLR) based on the first $n$ observations in stream $i$, i.e., 
\begin{equation*} \label{LLR_k_full}
    \lambda_i(n) := \sum_{t=1}^n \log \left( \frac{g_i(X_i(t))}{f_i(X_i(t))} \right).
 \end{equation*}
Fix thresholds $a,b>0$, and  let $T_i^{\text{SPRT}}$ denote the number of observations until the LLR exits the interval $(-b,a)$, and let $D_i^{\text{SPRT}}$ indicate whether the upper or the lower threshold is crossed, i.e., 
\begin{equation} \label{local SPRT}
\begin{aligned} T_i^{\text{SPRT}} &:= \inf\{ n\geq 1: \lambda_i(n) \notin (-b,a) \}, \\
    D_i^{\text{SPRT}} &:= 1\{  \lambda_i(T_i^{\text{SPRT}}) \geq a \}.
\end{aligned}
\end{equation}
The  SPRT in \eqref{local SPRT} with thresholds $a=\log((1-\beta)/\alpha)$ and $b=\log((1-\alpha)/\beta)$ 
 achieves both infima exactly if  there is no overshoot over the thresholds, 
and at least to a first-order asymptotic approximation as $\alpha,\beta\to 0$ with thresholds $a=|\log\alpha|$ and $b=|\log\beta|$, i.e., 
\begin{equation*}
\begin{aligned}
    & \Exp_{i,1}[T_i^{\text{SPRT}}] 
    \sim \inf_{(T_i,D_i)\in\Delta_i(\alpha,\beta)} \Exp_{i,1}[ T_i ] 
    \sim \frac{|\log\alpha|}{I_i} \\
    & \Exp_{i,0}[T_i^{\text{SPRT}}] 
    \sim \inf_{(T_i,D_i)\in\Delta_i(\alpha,\beta)} \Exp_{i,0}[ T_i ] 
    \sim \frac{|\log\beta|}{J_i}.
\end{aligned}
\end{equation*}

\subsection{Familywise error rates}
We now consider  the  multi-stream formulation, for which we need to specify a random time,  
$T_{\text{stop}}$, at which sampling is terminated in all streams, 
and  the subset of streams that we  declare as signals at that time, $D(T_{\text{stop}})$.  The error constraints that we impose are that the probabilities of at least one type-I error and at least one type-II error do not exceed  $\alpha$ and $\beta$, respectively, where $\alpha,\beta\in(0,1)$ are user-specified levels.
If the  true subset of signals  is $B\subseteq[K]$, then at least one type-I error occurs when $ D(T_{\text{stop}}) \backslash B\neq\emptyset$ and at least one type-II error occurs when  $B\backslash D(T_{\text{stop}})\neq\emptyset$. 
Therefore, our error constraints can be expressed as follows: 
\begin{align} \label{error}
    \Pro_B( T_{\text{stop}}<\infty) = 1, \quad  
    \Pro_B( D(T_{\text{stop}}) \backslash B\neq\emptyset)\leq\alpha, \quad \Pro_B(B\backslash  D(T_{\text{stop}})\neq\emptyset)\leq \beta, \quad 
    \forall \;   B\subseteq[K].
\end{align}  

\subsection{Sequential and active sampling }
We consider a sequential setup with active sampling, where at each time instant we observe only one stream, and this stream  is selected based on the already collected observations. For each  time instant $n\geq 1$, we denote by  $\cS(n)\in[K]$  the
 identity of the stream we observe at time $n$. We assume that, for each $n\geq 0$, $\cS(n+1)$ is $\cF^\cS(n)$-measurable,  where $\cF^\cS(0)$ is a $\sigma$-algebra independent of the observations and 
$$\cF^\cS(n) := \sigma\left(\cF^\cS(n-1), X_{\cS(n)}(n)\right),$$
and we refer to the sequence $\cS:=\{\cS(n):n\geq 1\}$ as a \textit{sampling rule}. Given such a sampling rule, 
 the termination time, $T_{\text{stop}}$, must be a stopping time  with respect to the filtration  $\{\cF^{S}(n): n\geq 0\}$, and the final decision, $D(T_{\text{stop}})$, must be $\cF^S(T_{\text{stop}})$-measurable. 
Temporarily,  we refer to such a tuple, $\delta=(S,T_{\text{stop}},D(T_{\text{stop}}))$, as a policy, and
we  denote by $\Delta$ the family of all 
 such  policies
and by $\Delta(\alpha,\beta)$ the subfamily of policies 
that satisfy the error constraints in \eqref{error}.
In Section \ref{subsec: asynchronous}, we will extend this definition to allow for early declaration of signals.

\subsection{Minimizing the expected total sample size}
We denote by $\cJ_B(\alpha, \beta)$ the 
optimal expected total sample size in $\Delta(\alpha,\beta)$  when the true subset of signals  is 
$B \subseteq[K]$, that is, 
\begin{equation} \label{objective funcs, Tstop} 
\cJ_B(\alpha, \beta):= \inf_{\delta\in\Delta(\alpha,\beta)}\Exp_B\left[ T_{\text{stop}}\right].
\end{equation}
This has been characterized
to a first-order asymptotic approximation as $\alpha, \beta \to 0$
in \cite{Aris_IEEE} 
and has been shown to be achieved in this asymptotic sense simultaneously for every $B \subseteq [K]$  by a policy that generalizes Wald's Sequential Probability Ratio Test (SPRT). 
To be specific,  for any sampling rule $\cS$, we refer to 
\begin{equation} \label{LLR_k}
\lambda_i^\cS(n) := \sum_{t=1}^n \log \left( \frac{g_i(X_i(t))}{f_i(X_i(t))} \right) \,  1\{\cS(t)=i\},
\end{equation}
as the local log-likelihood ratio statistic (LLR) in stream $i$ up to time $n$ under the sampling rule $\cS$ (for a formal justification for this terminology, see Appendix \ref{appendix, llr}). 
We refer to a policy
$(S,\hat T_{\text{stop}},\hat D(\hat T_{\text{stop}}))$
as an \textit{active SPRT} if 
(i)  a stream is no longer sampled once its LLR  has exited the interval $(-b,a)$, i.e., 
\begin{align}
\label{sensible} 
S(n) \in  \hat A(n-1) := \{i\in[K]: \lambda_i^{\cS}(n-1) \in (-b,a) \}, \quad  \forall \, n\in\bN,
\end{align}
where $-b<0<a$ are fixed thresholds, 
and (ii) the streams that are declared as signals are those whose LLRs have exceeded the upper threshold, $a$, i.e., 
\begin{equation} \label{temporary D} 
\begin{aligned} 
    \hat D(\hat T_{\text{stop}}) 
    &:= \{i \in [K]: \lambda_i^\cS(\hat T_{\text{stop}})\geq a\},
\end{aligned}
\end{equation}
where $\hat T_{\text{stop}}$ is the time at which all LLRs have been outside $(-b,a)$, i.e., 
\begin{equation} \label{general Tstop}
    \hat T_{\text{stop}} := \inf\{n\geq 1: \lambda_i^\cS(n)\notin(-b,a) \text{ for all } i\in[K]\}.
\end{equation}
It was shown in \cite{Aris_IEEE} that  if the thresholds are selected as $a=|\log \alpha|+\log K$ and  $b=|\log \beta|+\log K$, then the error constraints in \eqref{error} are satisfied and, with an appropriate design of the sampling rule $S$, the optimal expected total sample size in \eqref{objective funcs, Tstop} is also achieved to a first-order asymptotic approximation as $\alpha,\beta\to 0$.

Our first contribution in this work is that we  show that this first-order asymptotic optimality 
does not require any particular design of the sampling rule. 
More precisely, we show that  every active SPRT, defined by \eqref{sensible}-\eqref{general Tstop},
asymptotically minimizes the expected total sample size. 
Thus, an additional criterion needs to be introduced in order to distinguish among the different sampling rules that can be used as inputs in an active SPRT. 
As we discuss next, the criterion we adopt in this work is the early detection of individual signals. 

\begin{remark}
The problem of minimizing the expected total sample size in \eqref{objective funcs, Tstop} alone has been studied in sequential
multi-stream testing under more general active sampling setups, where more than one stream may be observed at each
time instant. Examples include \cite{Kobi_2015_active, Kobi_2018_heterogeneous, Kobi_2020_composite, Kobi_2022_switching, Aris_IEEE}, which incorporate prior information on the true subset of signals,
\cite{Kobi_2023_hiera, ISIT2026_CJ_CS}, which impose certain structure among the testing problems, and \cite{Aris_TIT2025}, which tolerates a small number of errors.
\end{remark}

\subsection{Asynchronous decisions} 
\label{subsec: asynchronous}

In the formulation considered so far, decisions for all streams are made upon termination, through the $\cF^{\cS}(T_{\text{stop}})$-measurable decision rule $D(T_{\text{stop}})$.
For an active SPRT, however, the decision in each stream can be effectively made earlier, 
namely, when that stream is sampled for the last time. 
More generally, in practice it is often desirable, or even critical, to make a decision in a stream 
as soon as sufficient evidence has been accumulated, so that appropriate actions can be taken promptly. 
Moreover, in many applications it is more important to identify signals early than to identify noises early. 
Motivated by these considerations, we now extend the formulation to allow each stream to be declared as a signal at any time before termination. 


Formally,
for every $n\geq 1$
we  denote by  $D(n)$  the subset of streams that have been declared as signals by time $n$, and set $D(0):=\emptyset$.
For every $n\in\bN$, we require $D(n)$ to be $\cF^\cS(n)$-measurable and that 
\begin{equation}
\label{req}
D(n-1)\subseteq D(n),
\qquad
D(n)\setminus D(n-1)\subseteq\{\cS(n)\}.
\end{equation}
The first condition ensures that signal declarations are irreversible, 
whereas the second requires that at any time instant only the currently sampled stream can be declared as a signal. 
We refer to such a sequence $D:=\{D(n):n\in\bN\}$ as a \emph{decision rule} 
and extend the definition of a policy from $(\cS,T_{\text{stop}},D(T_{\text{stop}}))$ 
to $(\cS,T_{\text{stop}},D)$. 
Note that condition \eqref{req} implies that  $D(n)=
D(T_{\text{stop}})$ for $n>T_{\text{stop}}$. 

From now on, we denote by $\Delta$ the family of all (extended) policies, 
and by $\Delta(\alpha,\beta)$ the subfamily of policies that satisfy the error constraints in \eqref{error}.
Moreover, note that the definition of the active SPRT implies that its decision rule is 
\begin{equation}
\label{general D}
\begin{aligned}
\hat D(n)
&:= \{i\in[K]: \lambda_i^\cS(n)\geq a\},
\qquad \forall\, n\in\bN.
\end{aligned}
\end{equation}
Thus, in what follows, we refer to any policy  
$(S,\hat {T}_{\text{stop}},\hat{D})$ as an \textit{active SPRT} if $S$ satisfies \eqref{sensible} and $\hat {T}_{\text{stop}},
\hat{D}$ are defined by \eqref{general Tstop}, \eqref{general D}. 

\subsection{Minimizing the expected successive signal declaration times} 
We are now ready to 
formulate the optimization problem that we propose in this work.
For each $k\in [K]$, we denote by $T_k$ the time of the $k_{th}$ signal declaration, that is, 
\begin{align} \label{general Tk}
    T_k := \inf & \{n>   T_{k-1}: D(n) \neq D(n-1)\}, \quad k \in [K],
\end{align}
where $T_0:= 0$.
We adopt the convention that the infimum of the empty set is $+\infty$, so if fewer than $k$ signal declarations are made by $T_{\text{stop}}$, then $T_k:=+\infty$. 
Therefore, $T_k\wedge T_{\text{stop}}$ represents the time until \textit{either declaring the $k_{th}$ signal or terminating the policy with fewer than $k$ signals declared}.
The problem we consider is to design a policy that, in addition to \eqref{objective funcs, Tstop}, achieves asymptotically as $\alpha,\beta\to 0$, 
\begin{equation} \label{objective funcs, Tk}   \cJ_{B,k}(\alpha, \beta):= \inf_{\delta\in\Delta(\alpha,\beta)}\Exp_B\left[ T_k\wedge  T_{\text{stop}}\right], 
\end{equation}
simultaneously for every true subset of signals $B\subseteq[K]$ \textit{and for every $k \in [K]$}.
We will show that, 
with an appropriate choice of the sampling rule, $S$, the active SPRT,
 $(S, \hat T_{\text{stop}}, \hat D)$,  
achieves all these infima, 
to a first-order asymptotic approximation as $\alpha, \beta \to 0$, as long as $\beta$ does not go to $0$ much faster than $\alpha$, in the sense that $|\log\beta|= O(|\log\alpha|)$.   

\begin{remark}
If the goal is to detect the existence of signals,
then only $T_1\wedge T_{\text{stop}}$ needs to be specified.
In this context, the only relevant optimization problem is \eqref{objective funcs, Tk} with $k=1$ and a weaker accuracy requirement. 
This has been considered, e.g., in \cite{Lai2011quickest, Malloy_Nowak_2013, Fellouris2013unstructured, Ali2016quickest, Fellouris_noniid, chaudhuri2024joint}.
\end{remark}

\section{Universal lower bounds} \label{section: universal lower bound}
In this section we establish universal lower bounds for the infima in  \eqref{objective funcs, Tk}.
For any non-empty subset $B\subseteq[K]$, we let $I_{(1)}(B)\geq\cdots\geq I_{(|B|)}(B)$
denote the non-increasingly ordered signal KL-divergences in $\{I_i:i\in B\}$.
The non-asymptotic lower bounds are presented in Theorem \ref{theorem: ALB} and their asymptotic versions in Corollary \ref{corollary: ALB}.

\begin{theorem} \label{theorem: ALB}
    Fix $\alpha,\beta\in(0,1)$ such that $\alpha+\beta<1/(K+1)$, and $B\subseteq[K]$.
    For $1\leq k\leq |B|$ we have
    \begin{equation} \label{non-ALB for tauk}
        \cJ_{B,k}(\alpha,\beta) \geq \min\left\{ \sum_{i\in B} \frac{d(1-p_i,\alpha)}{I_i} : 0\leq p_i\leq 1-\alpha \;\;\forall\, i\in B, \text{ and } \sum_{i\in B} p_i\leq |B|-k(1-\alpha-\beta)  \right\},
    \end{equation}
    and for $|B|<k\leq K$ we have
    \begin{equation} \label{non-ALB for taus}
        \cJ_B(\alpha,\beta)\geq \cJ_{B,k}(\alpha,\beta)
        \geq
        \sum_{i\in B}\frac{d(1-\beta,\alpha)}{I_i}
        +\sum_{i\notin B}\frac{d(1-\alpha,\beta)}{J_i}.
    \end{equation}
\end{theorem}

\begin{IEEEproof}
See Appendix \ref{app: proof of the lower bound}.
\end{IEEEproof}

\begin{remark}
The  optimization problem in the right-hand-side of \eqref{non-ALB for tauk} is a convex problem with linear constraints. Its unique minimizer is 
\begin{equation} \label{optimizers}
\begin{gathered}
    p_i^* = \frac{1-\alpha}{1-\alpha+\alpha\exp\{\lambda\,  I_i\}}, \quad  i\in B, \\
    \text{with $\lambda$ being the solution to } \sum_{i\in B} p_i^* = |B|-k(1-\alpha-\beta).
\end{gathered}
\end{equation}   
The proof of this is contained in the proof of Theorem \ref{theorem: ALB}.
\end{remark}

Equation
\eqref{non-ALB for taus} states that 
the expected time until either declaring more signals than there actually exist or terminating is 
lower bounded by the sum of the expected sample sizes required for solving every testing problem with type-I error rate $\alpha$ and type-II error rate $\beta$.

On the other hand, equation \eqref{non-ALB for tauk} states that, for every $1\leq k\leq|B|$, the expected time until either declaring the $k_{th}$ signal or terminating is lower bounded by the sum of the expected sample sizes required for solving every testing problem in $B$ with type-I error rate $\alpha$ and type-II error rate $p_i$ in stream $i\in B$, where $\{p_i:i\in B\}$ are some dummy variables satisfying certain constraints.  
Intuitively, $p_i$ can be understood as the probability that true signal stream $i\in B$ 
is missed among the first $k$ signal declarations, i.e., an early-miss probability. 
The constraints on the $\{p_i:i\in B\}$ come from the actual familywise error constraints. 
Specifically, under these constraints, with probability at least $1-\alpha-\beta$ the final decision is correct. 
On this event, the first $k$ signal declarations must all correspond to true signal streams. 
Therefore, the expected number of true signal streams among the first $k$ signal declarations 
is at least $k(1-\alpha-\beta)$. 
Accordingly, the effective early-miss probabilities must satisfy
$\sum_{i\in B}(1-p_i)\geq k(1-\alpha-\beta)$, 
or, equivalently,
$\sum_{i\in B}p_i\leq |B|-k(1-\alpha-\beta)$.
Thus, the optimization in \eqref{non-ALB for tauk} can be viewed as allocating 
the allowable early-miss probabilities across the true signal streams 
so as to minimize the resulting total sample size requirement.
The optimizer $\{p_i^*:i\in B\}$ in \eqref{optimizers} gives the most favorable such allocation.
In particular, $p_i^*$ is decreasing in $I_i$. 
Thus, an easy signal stream, that is, a signal stream with a larger signal KL-divergence, is assigned with a smaller early-miss probability, 
or equivalently, is required to appear among the first $k$ signal declarations with a larger probability. 
This intuition  becomes  transparent in the asymptotic lower bound below.

\begin{corollary}\label{corollary: ALB}
Fix $B\subseteq[K]$. As $\alpha,\beta\to0$, for $1\leq k\leq |B|$ we have
\begin{equation} \label{coro: asym_lower_bound_k}
\cJ_{B,k}(\alpha,\beta)\gtrsim
\sum_{i=1}^k\frac{|\log\alpha|}{I_{(i)}(B)},
\end{equation}
and for $|B|<k\leq K$ we have
\begin{equation} \label{coro: asym_lower_bound}
\cJ_B(\alpha,\beta)\geq \cJ_{B,k}(\alpha,\beta)
\gtrsim
\sum_{i\in B}\frac{|\log\alpha|}{I_i}
+\sum_{i\notin B}\frac{|\log\beta|}{J_i}.
\end{equation}
\end{corollary}
\begin{IEEEproof}
See Appendix \ref{app: proof of the lower bound}.
\end{IEEEproof}

The asymptotic lower bound in \eqref{coro: asym_lower_bound_k} provides a particularly simple relaxation of the finite-sample lower bound in \eqref{non-ALB for tauk}. 
Specifically, for every $1\leq k\leq |B|$, to first order, the requirement of being declared as signals among the first $k$ declarations concentrates on the $k$ easiest signal streams.
Consequently, the expected time until either the $k_{th}$ signal declaration or termination is lower bounded by the sum of the expected sample sizes required for solving these $k$ testing problems.
This suggests that an asymptotically optimal sampling rule should prioritize sampling from signal streams 
with larger signal KL-divergences, 
which provides the main intuition for the design of our sampling rule in the next section.

\section{Asymptotic optimality of the generic active SPRT} \label{section: active SPRT}
Recall that we refer to a policy an active SPRT if 
sampling from a stream is terminated as soon as its LLR is outside an interval $(-b,a)$, and the stream is declared as a signal if its LLR crosses the upper threshold, i.e., if its sampling rule satisfies \eqref{sensible} and its termination time and decision rule are given by \eqref{general Tstop} and \eqref{general D}, respectively. 
In this section we show that  such a policy  minimizes,
to a first-order asymptotic approximation, the expected total sample size, as well as the expected time until either declaring more signals than there actually exist or terminating, irrespectively of the specific choice of the sampling rule. 

To show this, we  observe that condition \eqref{sensible} on the sampling rule, together with the assumption of independence over time and across streams,  imply that (i) the number of observations from stream $i$ until termination,
$\sum_{t=1}^  {\hat T_{\text{stop}}} 
  1\{S(t)=i\}$,
has the same distribution as $T_i^{\text{SPRT}}$,
defined in \eqref{local SPRT}, and (ii)  the  probability of declaring stream $i$ as a signal or a noise is the same as that of $D_i^{\text{SPRT}}$, i.e., 
\begin{equation*}
\begin{aligned}
    \Pro_B(i\in \hat D(\hat T_{\text{stop}})) & = \Pro_B(D_i^{\text{SPRT}}=1), \\
    \Pro_B(i\notin \hat D(\hat T_{\text{stop}})) & = \Pro_B(D_i^{\text{SPRT}}=0).
\end{aligned}
\end{equation*}
These observations provide the basis for the following result.

\begin{theorem} \label{theorem, error control of the general decision rule}
Consider the active SPRT, $(S,\hat T_{\text{stop}},\hat D)$, defined by \eqref{sensible}, \eqref{general Tstop} and \eqref{general D}.
For any $a,b>0$ and $B\subseteq[K]$ we have
    \begin{align} \label{a.s. finiteness and error control}
    \begin{split}
        \Pro_B(\hat T_{\text{stop}}<\infty) &= 1, \\
        \Pro_B(\hat D(\hat T_{\text{stop}})\backslash B\neq\emptyset) &\leq (K-|B|) \ e^{-a}, \; \\
        \Pro_B(B\backslash \hat D(\hat T_{\text{stop}})\neq\emptyset) &\leq |B|\, e^{-b}.
    \end{split}
    \end{align}
Thus, for any $\alpha,\beta\in(0,1)$, 
the policy belongs to $\Delta(\alpha,\beta)$ when 
\begin{equation} \label{selection of thresholds}
a=|\log\alpha|+\log K, \quad b=|\log\beta|+\log K.
\end{equation}
\end{theorem}

\begin{IEEEproof} 
    Fix  $a,b>0$ and $B\subseteq[K]$.
    If $\hat T_{\text{stop}}=\infty$, then there must exist a stream  $i\in[K]$ from which the number of samples is  infinite, but whose LLR is always between $(-b,a)$.
    The probability of this event is zero, since the LLR of every stream has a non-zero drift.     
    Thus, we have the first line of \eqref{a.s. finiteness and error control}. Next, we only show the second line of \eqref{a.s. finiteness and error control} as the third line is similar. 
    It is clear that $\{\hat D(\hat T_{\text{stop}})\backslash B\neq\emptyset\}=\{\exists\,i\notin B: i\in\hat D(\hat T_{\text{stop}})\}$.  
    Thus, by the union bound, 
    \begin{equation*}
    \begin{aligned}
        \Pro_B(\hat D(\hat T_{\text{stop}})\backslash B\neq\emptyset) \leq \sum_{i\notin B} \Pro_B(i\in\hat D(\hat T_{\text{stop}}))
        = \sum_{i\notin B} \Pro_B(D_i^{\text{SPRT}}=1) \leq (K-|B|) e^{-a},
    \end{aligned}
    \end{equation*}
    where the last inequality is a well-known bound for the type-I error rate of the SPRT (see, e.g., \cite[Chapter 3]{Tartakovsky_Book}).
\end{IEEEproof}

We proceed with an upper bound on the moments of $\hat T_{\text{stop}}$, for which we need to introduce some additional notations. 
For any $B\subseteq[K]$ and $\epsilon\in(0,1)$ we set
\begin{equation} \label{cV's}
    \cV_B(\epsilon) := \sum_{i\in B}\cV_i^+(\epsilon) + \sum_{i\notin B}\cV_i^-(\epsilon),
\end{equation}
where,  for any $i\in[K]$,
\begin{equation} \label{cVi's}
\begin{aligned}
    \cV_i^+(\epsilon) & := \sup\left\{n\geq 1: \lambda_i(n)/n\leq I_i(1-\epsilon) \right\} + 1, \\
    \cV_i^-(\epsilon) & := \sup\left\{n\geq 1: -\lambda_i(n)/n\leq J_i(1-\epsilon) \right\} + 1.
\end{aligned}
\end{equation}
Note that these random times are not stopping times. 
Also note that, for any $i\in B$,
\begin{equation*}
\begin{aligned}
    \{ \lim_{n\to\infty} \lambda_i(n)/n = I_i \} 
    & = \{ \sup\left\{n\geq 1: |\lambda_i(n)/n-I_i|>\epsilon \right\}<\infty \text{ for all } \epsilon>0 \} \\
    & \subseteq \{ \cV_i^+(\epsilon)<\infty \text{ for all } \epsilon\in(0,1) \},
\end{aligned}
\end{equation*}
and, similarly, for any $i\notin B$,
$$\{ \lim_{n\to\infty} -\lambda_i(n)/n = J_i \} \subseteq \{ \cV_i^-(\epsilon)<\infty \text{ for all } \epsilon\in(0,1) \}.$$
Therefore, by the Strong Law of Large Numbers, condition \eqref{only assumption} implies $\Pro_B(\cV_B(\epsilon)<\infty)=1$ for all $\epsilon\in(0,1)$.
In Lemma \ref{Lemma: properties of LLR} we further show that condition \eqref{only assumption}  implies 
$$V_B(r,\epsilon):=\Exp_B[(\cV_B(\epsilon))^r]<\infty \; \forall \; r\geq 1 \text{ and } \epsilon\in(0,1).$$

\begin{theorem} \label{theorem, upper bound on Tstopr}
Consider the active SPRT in Theorem \ref{theorem, error control of the general decision rule}.
Fix $B\subseteq[K]$, $r\geq 1$ and $\epsilon\in(0,1)$. 
For any $a,b>0$ we have
\begin{equation} \label{upper bound on Tstopr}
\begin{aligned}
    \Exp_B[(\hat T_{\text{stop}})^r]
    \leq
    2^{r-1} \left( \left(\frac{1}{1-\epsilon} \left( \sum_{i\in B} \frac{a}{I_i} + \sum_{i\notin B} \frac{b}{J_i} \right) \right)^r + V_B(r,\epsilon) \right).
\end{aligned}
\end{equation}
Thus, as $a,b \to \infty$ we have
\begin{equation} \label{E[r]1/r}
    (\Exp_B[( \hat T_{\text{stop}})^r])^{1/r}= O(a\vee b) \text{ for all } r\geq 1.
\end{equation}
\end{theorem}

\begin{IEEEproof}
First, we observe that the total number of observations until termination,
\begin{equation*}
   \hat T_{\text{stop}} = \sum_{i\in[K]} \sum_{t=1}^{\hat T_{\text{stop}}} 1\{S(t)=i\},
\end{equation*}
has the same distribution as $\sum_{i\in[K]} T_i^{\text{SPRT}}$, where $T_i^{\text{SPRT}}$ is defined in \eqref{local SPRT}.   Thus, 
$$
\Exp_B[(\hat T_{\text{stop}})^r] = \Exp_B\left[ \left(\sum_{i\in[K]} T_i^{\text{SPRT}}\right)^r \right].
$$
According to Lemma \ref{Lemma: for AUB},
for every $\epsilon\in(0,1)$ we have 
\begin{equation*}
\begin{aligned}
    T_i^{\text{SPRT}} \leq
    \begin{cases}
    \begin{aligned}
        & \frac{a}{I_i(1-\epsilon)} + \cV_i^+(\epsilon), && \text{for } i\in B, \\
        & \frac{b}{J_i(1-\epsilon)} + \cV_i^-(\epsilon), && \text{for } i\notin B, 
    \end{aligned}
    \end{cases}
\end{aligned}
\end{equation*}
so 
\begin{equation} \label{sum of TiSPRT}
    \sum_{i\in[K]} T_i^{\text{SPRT}} \leq \frac{1}{1-\epsilon} \left( \sum_{i\in B} \frac{a}{I_i} + \sum_{i\notin B} \frac{b}{J_i} \right) + \cV_B(\epsilon),
\end{equation}
where $\cV_i^+(\epsilon),\cV_i^-(\epsilon)$, $\cV_B(\epsilon)$ are defined in \eqref{cV's}-\eqref{cVi's}.
The desired result follows by the  $C_r$-inequality, which states that 
$\Exp[(X+Y)^r] \leq 2^{r-1}(\Exp[X^r]+\Exp[Y^r])$
for any non-negative random variables $X,Y$ and $r\geq 1$. 
For the second part of the theorem, 
note that $V_B(r,\epsilon)$ is finite for all $r\geq 1$ and $\epsilon\in(0,1)$ based on Lemma \ref{Lemma: properties of LLR}.
Letting first $a,b\to\infty$ and then  $\epsilon\to 0$ in \eqref{upper bound on Tstopr}, we have
\begin{equation*}
    (\Exp_B[(\hat T_{\text{stop}})^r])^{1/r} \lesssim 2^{1-1/r} \left( \sum_{i\in B} \frac{a}{I_i} +  \sum_{i\notin B} \frac{b}{J_i} \right)= O(a\vee b).
\end{equation*}
\end{IEEEproof}
    
\begin{corollary}
Consider the active SPRT in Theorem \ref{theorem, error control of the general decision rule}.
Suppose that the thresholds $a,b$ are selected so that the policy belongs to $\Delta(\alpha,\beta)$ for all $\alpha,\beta\in(0,1)$ and $a\sim|\log\alpha|$, $b\sim|\log\beta|$ as $\alpha,\beta\to 0$, e.g., as in \eqref{selection of thresholds}. Then for any $B\subseteq[K]$ and  
$|B|< k\leq K$, as $\alpha,\beta\to 0$ we have
\begin{equation} \label{AO for k geq |B|}
\begin{aligned} 
     \Exp_B[  \hat T_k\wedge  \hat T_{\text{stop}}] & \sim \cJ_{B,k} (\alpha, \beta) \sim   \cJ_{B} (\alpha, \beta) \sim \Exp_B[  \hat T_{\text{stop}}] \\
      &\sim  \; 
    \sum_{i\in B} \frac{|\log\alpha|}{I_i} +\sum_{i\notin B} \frac{|\log\beta|}{J_i}.
\end{aligned}   
\end{equation}
\end{corollary}
\begin{IEEEproof}
Setting $r=1$ in the previous theorem and comparing with the asymptotic lower bound in \eqref{coro: asym_lower_bound} proves \eqref{AO for k geq |B|}.
\end{IEEEproof}



\section{Quick detection of signals} \label{section: quick detection of signals}
In the previous section, we showed that an active SPRT asymptotically minimizes 
the expected total sample size and the expected time until declaring more signals 
than there actually exist or terminating, essentially regardless of the choice of the sampling rule.
Throughout this section, we therefore restrict attention to active SPRTs, i.e., the termination time and the decision rule are fixed as in \eqref{general Tstop} and \eqref{general D}, 
and focus on the design of the sampling rule for the early detection of individual signals.
Before introducing our proposed sampling rule, 
we first consider two alternative sampling rules that offer important insights.

In what follows, we recall that the signal KL-divergence of stream $i$, $I_i$, characterizes the detection difficulty of stream $i$ when it is a signal. 
Without loss of generality,  we assume that the streams are indexed so that
$I_1\geq I_2\geq \cdots \geq I_K$.

\subsection{Sampling in tandem} \label{subsec: sampling-in-tandem}
The most straightforward way of sampling is to always sample the undecided stream with the smallest index, i.e., 
\begin{equation*}
   S(n) := \min \{ \hat A(n-1)\} \text{ for } n \geq 1,
\end{equation*}
where 
$\hat A(n-1)$ is the subset of undecided streams at time $n$, defined in \eqref{sensible}. 
In other words, 
keep sampling stream $1$ until  its LLR exits $(-b,a)$, then do the same for stream $2$, etc. 
This approach will indeed be efficient  
when the true signals are the ones with the largest signal KL-divergences, 
i.e., when $B=\{1, \ldots, |B|\}$.  However, it will be very  inefficient  otherwise. 
For example, if $B=\{K\}$, then all noise streams 
will be identified before even starting to sample the  unique signal stream.

\subsection{Following the leader} \label{subsec: follow-the-leader}
The previous sampling approach  treats signals and noises equally, and as a result it may spend  a lot of time in identifying noises. A sampling approach that   prioritizes  (the unknown subset of) signals is to always  sample the 
undecided stream with the largest LLR, i.e.,  
\begin{equation} \label{cohen's sampling rule}
   S(n) := \argmax_{i\in \hat A(n-1)} \lambda_i^{S}(n-1) \text{ for } n> 1.
\end{equation}
When there are ties, we sample the stream with the smallest index.
The policy $(S, \hat T_{\text{stop}}, \hat D)$ has been proposed in \cite[Section IV.B]{Kobi_2015_active}, and we refer to it as the \emph{follow-the-leader policy}.
The following theorem and corollary state that, with this sampling rule,  
the active SPRT  asymptotically minimizes the expected time until either the \emph{$|B|_{th}$ signal declaration} or termination.

\begin{theorem} \label{theorem, cohen's}
    Consider the follow-the-leader policy, $(S, \hat T_{\text{stop}}, \hat D)$, with $S$ defined in \eqref{cohen's sampling rule} and $\hat T_{\text{stop}},\hat D$ in \eqref{general Tstop}, \eqref{general D}. 
    Then, for any non-empty set $B\subseteq[K]$ we have
    \begin{align} \label{AUB, cohen's proc}
        \Exp_B[\hat T_{|B|}\wedge\hat T_{\text{stop}}] & \lesssim \sum_{i\in B} \frac{a}{I_i}, 
    \end{align}
    as $a,b\to\infty$ so that $b= O(a)$. 
\end{theorem}
\begin{IEEEproof}
    See Appendix \ref{app: proof of the upper bound}.
\end{IEEEproof}

Combining Theorem \ref{theorem: ALB} and \ref{theorem, cohen's}, the next corollary follows.
\begin{corollary} \label{corollary, cohen's}
    Consider the follow-the-leader policy in Theorem \ref{theorem, cohen's}.
    Suppose that the thresholds $a,b$ are selected so that the policy belongs to $\Delta(\alpha,\beta)$ for all $\alpha,\beta\in(0,1)$ and $a\sim|\log\alpha|$, $b\sim|\log\beta|$ as $\alpha,\beta\to 0$, e.g., as in \eqref{selection of thresholds}. 
    Then, for any non-empty $B\subseteq[K]$ we have
    \begin{align*}
        \Exp_B[ \hat T_{|B|}\wedge \hat T_{\text{stop}}] & \sim \cJ_{B, |B|}(\alpha, \beta)  \sim  \sum_{i\in B} \frac{ |\log\alpha|}{I_i}
    \end{align*}
    as $\alpha,\beta\to 0$ so that $|\log\beta|= O(|\log\alpha|)$.
\end{corollary}

\begin{IEEEproof}
    Compare the asymptotic upper bound in \eqref{AUB, cohen's proc} with the asymptotic lower bound in \eqref{coro: asym_lower_bound_k}. \\
\end{IEEEproof}

However, the follow-the-leader policy is prone to overlooking some easy signals 
and can therefore be inefficient for their early detection.
Indeed, during its first traversal of the streams, it switches away from the current stream 
as soon as its LLR becomes negative.
However, even if the current stream is a signal whose LLR has a large positive drift, this can occur after the first observation 
with a non-zero probability, independently of the thresholds.
Consequently, an easy signal may be deferred in favor of more difficult ones
before enough evidence has been accumulated to reveal its positive drift.

\subsection{The proposed sampling rule}
We have discussed so far two sampling rules that sit on two extremes. Sampling-in-tandem is completely non-adaptive, as it samples streams in a completely pre-determined order, whereas following-the-leader is too  greedy, and as a result vulnerable to missing easily  detectable signals. The  sampling approach we propose in this work adds adaptivity to sampling-in-tandem and patience to following-the-leader. 
To be specific,  as in the case of sampling-in-tandem, 
 we  first sample stream $1$  until its LLR either exceeds the detection threshold $a>0$  or is below a non-positive threshold $-b'\leq 0$, less extreme than the one used for the declaration of noises, i.e., $b'\leq b$.  In the former case stream $1$ is declared as a signal, and in the latter, no call is made. We then  repeat the same process for streams $2,\ldots,K$. Formally, the proposed sampling rule is defined as
\begin{equation} \label{proposed sampling rule, Phase I, 1}
    \hat\cS(n) := i \text{ for } i\in[K] \text{ and } \tau_{i-1} < n\leq \tau_i ,
\end{equation}
where $\tau_0 := 0$ and
\begin{equation} \label{proposed sampling rule, Phase I, 2}
    \tau_i:= \inf\{ n> \tau_{i-1}: \lambda_i^{\hat\cS}(n)\notin(-b',a) \} \text{ for } i\in[K].
\end{equation}
We denote this sampling rule by $\hat S$ and the resulting policy by $(\hat S,\hat T_{\text{stop}}, \hat D)$. 

Note that equations \eqref{proposed sampling rule, Phase I, 1}-\eqref{proposed sampling rule, Phase I, 2} specify the sampling rule up to the random time $\tau_K$. At this time, each LLR is either above $a$ or  below $-b'$. In the former case,  the  stream has been declared as a signal, and in the latter, it has been declared as a noise if and only if its LLR is also smaller than $-b$. 
Thus, sampling will need to continue for streams whose LLRs lie in $(-b,-b']$.
We  refer to $[0, \tau_K]$ as   Phase I and to  $[\tau_K+1, T_{\text{stop}}]$ as Phase II. 
We will establish the desired asymptotic optimality property, \textit{irrespectively of how sampling is performed in Phase II}, as long as the threshold $b'$ is selected carefully. To understand the role played by this tuning parameter, it is  interesting to consider two extreme cases.
If $b'=b$, then the proposed sampling rule 
reduces to sampling-in-tandem (Section \ref{subsec: sampling-in-tandem}).
If $b'=0$, then  the sampling rule in Phase I reduces to following-the-leader (Section \ref{subsec: follow-the-leader}). 
Therefore, a very large value of \(b'\) may lead to excessive sampling of noise streams, whereas a very small value may provide insufficient exploration of signal streams and cause some of the easy ones to be set aside prematurely. A moderate choice of \(b'\) strikes a balance between these two effects and facilitates the earlier detection of easier signals.

\begin{theorem} \label{theorem: AUB}
Consider the proposed policy, $(\hat S, \hat T_{\text{stop}}, \hat D)$, with $\hat S$ defined in  \eqref{proposed sampling rule, Phase I, 1}-\eqref{proposed sampling rule, Phase I, 2} during Phase I and satisfying \eqref{sensible} during Phase II, and $\hat T_{\text{stop}},\hat D$ defined in \eqref{general Tstop}, \eqref{general D}.
Then, for any $B\subseteq[K]$, $1\leq k\leq |B|$ and $\epsilon\in(0,1)$ we have
    \begin{equation} \label{non-asy upper bound on EB[Tk]}
    \begin{aligned}
        \Exp_B[ \hat T_k\wedge  \hat T_{\text{stop}}] & \leq \frac{1}{1-\epsilon} \sum_{i=1}^k \frac{a}{I_{(i)}(B)} + V_B(1,\epsilon) 
        + O(b') + O(a\vee b) \cdot O(e^{-\frac{1}{2} (a\wedge b')}).
    \end{aligned}
    \end{equation}
\end{theorem}
\begin{IEEEproof}
    See Appendix \ref{app: proof of the upper bound}.
\end{IEEEproof}

\begin{remark}
   If $b= O(a)$, then
    the higher-order term in \eqref{non-asy upper bound on EB[Tk]}, omitting additive and multiplicative constants, is $b'+a e^{-b'/2}$.
    Balancing these two terms suggests the choice
    of $b'=\Theta(\log a)$.
\end{remark}

Combining Theorem \ref{theorem: ALB} and \ref{theorem: AUB}, next corollary follows.

\begin{corollary} \label{corollary, AO}
Consider the proposed policy in Theorem \ref{theorem: AUB}.
Suppose that the thresholds $a,b$ are selected so that the policy belongs to $\Delta(\alpha,\beta)$ for all $\alpha,\beta\in(0,1)$ and $a\sim|\log\alpha|$, $b\sim|\log\beta|$ as $\alpha,\beta\to 0$, e.g., as in \eqref{selection of thresholds}. 
If $b'$ is selected so that $b'\to\infty$ and $b'= o(a)$ as $\alpha,\beta\to 0$, then for any $B\subseteq[K]$ and  $1\leq k\leq |B|$ we have
\begin{align*}
    \Exp_B[ \hat T_k\wedge \hat T_{\text{stop}}] & \sim \cJ_{B,k}(\alpha, \beta) \sim  \sum_{i=1}^k \frac{ |\log\alpha|}{I_{(i)}(B)}
\end{align*}
as $\alpha,\beta\to 0$ so that $|\log\beta|= O(|\log\alpha|)$. 
\end{corollary}

\begin{IEEEproof} 
Fix  $B\subseteq[K]$ and  $1\leq k\leq |B|$.
In \eqref{non-asy upper bound on EB[Tk]}, first $a,b,b'\to\infty$  so that $b= O(a)$ and $b'= o(a)$ and then letting $\epsilon\to 0$, we have \begin{align*}
      \Exp_B[\hat T_k\wedge  \hat T_{\text{stop}}] & \lesssim \sum_{i=1}^k \frac{a}{I_{(i)}(B)}.   
  \end{align*}
    Comparing  with the asymptotic lower bounds in  \eqref{coro: asym_lower_bound_k} completes the proof. 
\end{IEEEproof}

\begin{algorithm}
    \caption{The proposed policy} \label{the proposed algorithm}
    \begin{algorithmic}[1] 
        \State Input $a,b>0$ and $0\leq b'\leq b$
        
        \State Initialize $k=0$, $n=0$, $D=\emptyset$, $N=\emptyset$, 
        \text{and}
        $\lambda_i=0$ for $1\leq i\leq K$  
        \For{$1\leq i\leq K$} \hfill \textit{// Phase 1}
            
            \While{$\lambda_i\in(-b',a)$}
                \State $n=n+1$
                
                \State generate new $X_i$
                
                \State compute $\lambda_i=\lambda_i+\log\left(\frac{g_i(X_i)}{f_i(X_i)}\right)$ 
            \EndWhile
            \If{$\lambda_i\geq a$}
                \State $k=k+1$, $T_k=n$, $D=D\cup\{i\}$
            \EndIf
            \If{$\lambda_i\leq-b$}
                \State $N=N\cup\{i\}$
            \EndIf
        \EndFor
        \While{$D\cup N\neq[K]$} \hfill \textit{// Phase II}
            \State $n=n+1$

            \State let $i$ be one of 

            \vspace{2mm}
            $\begin{cases}
            \begin{alignedat}{2}
                & \argmax_{i'\in[K]\backslash(D\cup N)} \lambda_{i'} \quad & \textit{// Follow-the-leader} \\
                & \argmax_{i'\in[K]\backslash(D\cup N)} |\lambda_{i'}| \quad & \textit{// Follow-the-absolute-leader} \\
                & \min \{[K]\backslash(D\cup N)\} \quad & \textit{// Sample-in-tandem} 
            \end{alignedat}
            \end{cases}$
            \vspace{2mm}

            \State generate new $X_i$
            
            \State compute $\lambda_i=\lambda_i+\log\left(\frac{g_i(X_i)}{f_i(X_i)}\right)$ 
            \If{$\lambda_i\geq a$}
                \State $k=k+1$, $T_k=n$, $D=D\cup\{i\}$
            \EndIf
            \If{$\lambda_i\leq -b$}
                \State $N=N\cup\{i\}$
            \EndIf
        \EndWhile
        \State $T_{\text{stop}} = n$
        \State Output $T_1,\ldots,T_k$, $T_{\text{stop}}$, $D$
    \end{algorithmic}
\end{algorithm}

We stress that the desired asymptotic optimality is achieved through a careful design of the sampling rule in Phase I, whereas no requirement is imposed on the sampling rule in Phase II apart from \eqref{sensible}. Nevertheless, the choice of the sampling rule in Phase II can affect the finite-sample performance. 
Since we are interested in the quick detection of all signals, a natural choice is to follow the leader, i.e.,
\begin{equation*}
   \hat\cS(n)
   :=
   \argmax_{i\in \hat A(n-1)}
   \lambda_i^{\hat\cS}(n-1),
   \qquad n>\tau_K .
\end{equation*}
Other Phase-II sampling rules may be preferable under different practical objectives. 
For example, with the aim of quick identification of all streams, 
one may sample the undecided stream whose LLR has the largest absolute value, i.e.,
\begin{equation*}
    \hat\cS(n)
    :=
    \argmax_{i\in \hat A(n-1)}
    \big|\lambda_i^{\hat\cS}(n-1)\big|,
    \qquad n>\tau_K .
\end{equation*}
If reducing the number of switches is important, one may sample in tandem, i.e.,
\begin{equation*}
    \hat\cS(n)
    :=
    \min \{\hat A(n-1)\},
    \qquad n>\tau_K .
\end{equation*}
A pseudocode of the complete policy is presented in Algorithm \ref{the proposed algorithm}.

\section{Numerical studies} \label{section: numerical studies}
In this section, we present numerical studies.
First, we visualize the effect of the tuning parameter $b'$ on the performance of the proposed policy. 
Then, we compare the performance of the proposed policy with the ``follow-the-leader" policy in Section \ref{subsec: follow-the-leader} and the oracle policy that applies an SPRT with thresholds $a,b$ to the streams in the correct order, i.e., first signals in the non-increasing order of their signal KL-divergences, then noises in any order.

The data model is as follows:
For each $i\in[K]$, we assume that $X_i$ are i.i.d. Gaussian, with unknown mean $\mu_i$ and unit variance, and the testing problem of interest is 
\begin{equation*}
    \mu_i=0 \quad \text{versus} \quad
    \mu_i=\delta_i,
\end{equation*}
for some $\delta_i>0$.
We fix $K=10$, $\boldsymbol{\delta}=(\delta_1,\ldots,\delta_{10})=(1.5,1.5,1.25,1.25,1,1,0.75,0.75,0.5,0.5)$, and we set the true, unknown subset of signals as $B=\{2,4,6,8,10\}$.
Note that the streams are indexed in a non-increasing order of their KL-divergences and, within each pair of streams with the same KL-divergence, the true signal appears after the true noise, which is a more difficult setup for the two practical policies.
We always use equal thresholds $a=b$.

\begin{figure}
    \centering
    \includegraphics[width=0.4\linewidth]{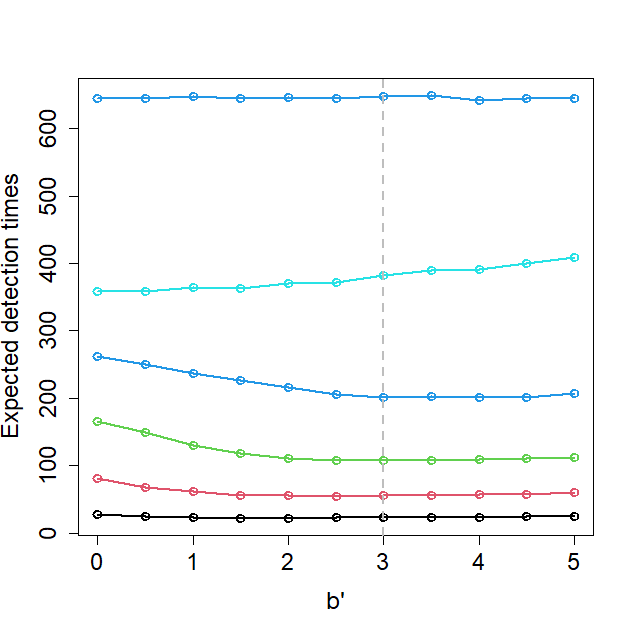}
    \caption{
    Expected detection times of the proposed policy, i.e., $\Exp_B[\hat T_k\wedge \hat T_{\text{stop}}]$ for $k\in[K]$, against $b'$.
    Curves from bottom to top correspond to $k=1$,$\ldots$, $k=5$, and $k=6,7,8,9,10$ which basically overlap.
    The vertical, dashed, gray line: the value of $\log(a)$.
    }
    \label{fig: ESS_b'}
\end{figure}

\begin{figure}
    \centering
    \subfloat[{$\Exp_B[T_k\wedge T_{\text{stop}}]$ for $k=1$}]{
        \includegraphics[width=0.4\linewidth]{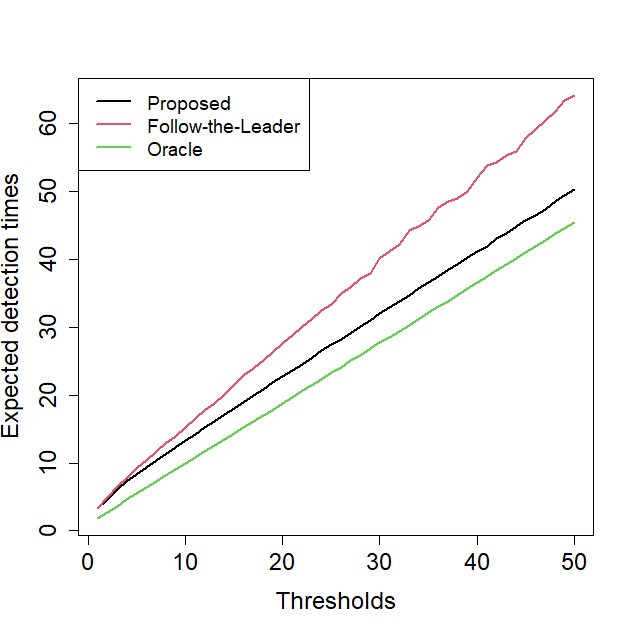}
    }\hspace{0.02\linewidth}
    \subfloat[$k=2$]{
        \includegraphics[width=0.4\linewidth]{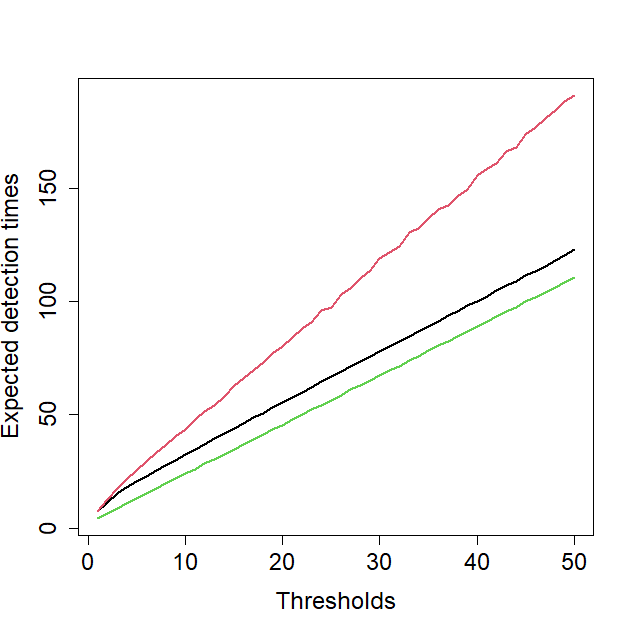}
    }\\[-1em]
    \subfloat[$k=3$]{
        \includegraphics[width=0.4\linewidth]{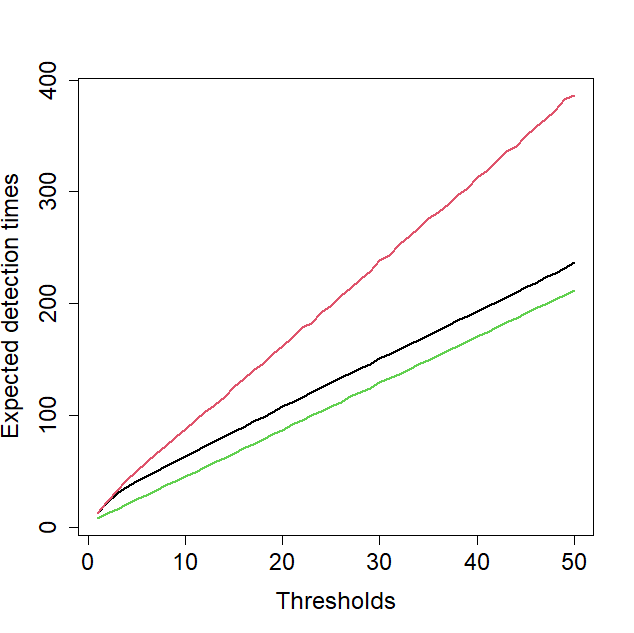}
    }\hspace{0.02\linewidth}
    \subfloat[$k=4$]{
        \includegraphics[width=0.4\linewidth]{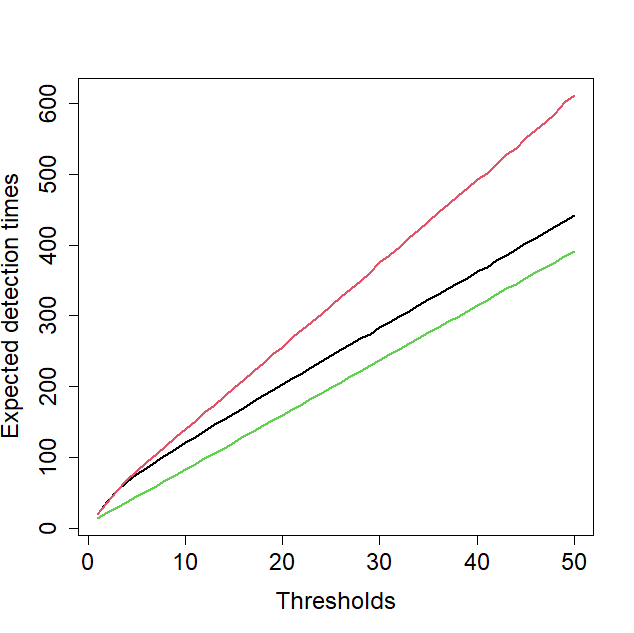}
    }\\[-1em]
    \subfloat[$k=5$]{
        \includegraphics[width=0.4\linewidth]{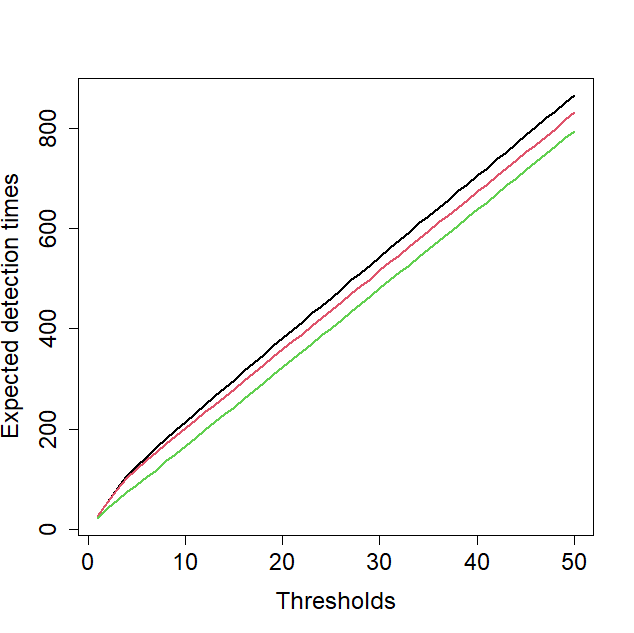}
    }\hspace{0.02\linewidth}
    \subfloat[{$k=6,\ldots,10$}]{
        \includegraphics[width=0.4\linewidth]{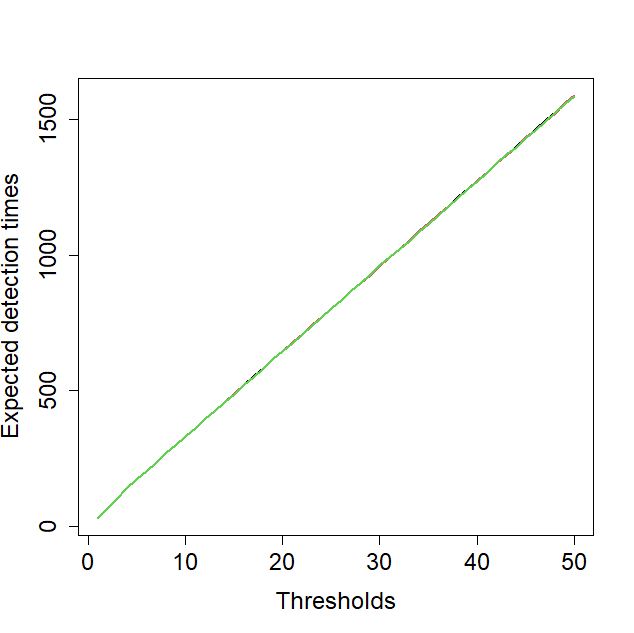}
    }
    \caption{
    Expected detection times of the three policies against thresholds $a=b$.
    The three lines correspond to the three policies, as shown in the legend of the first subfigure.
    }
    \label{fig: comparison}
\end{figure}

\begin{figure}
    \centering
    \includegraphics[width=0.4\linewidth]{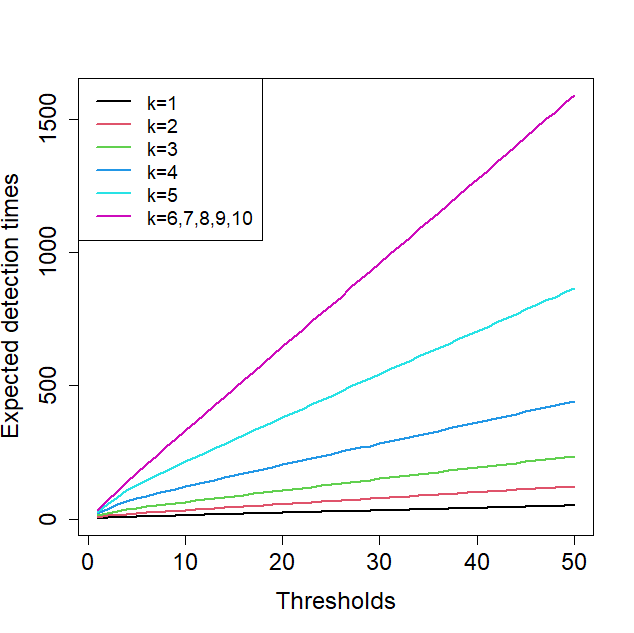}
    \includegraphics[width=0.4\linewidth]{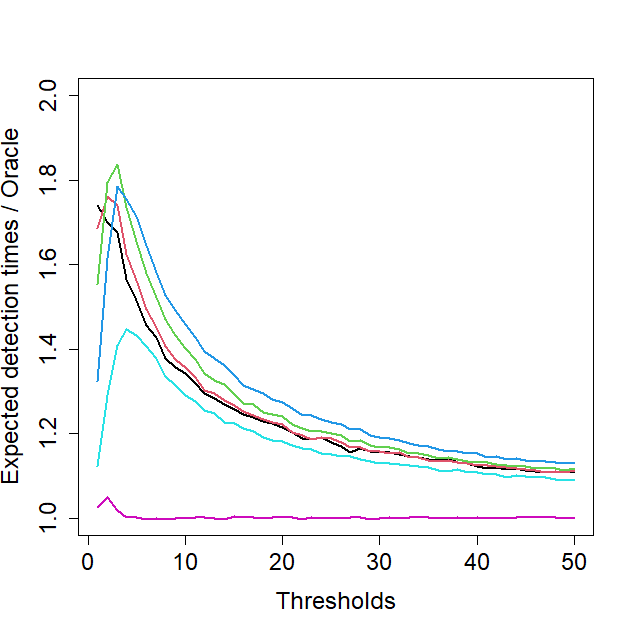} 
    \caption{
    Expected detection times of the proposed policy (left) and expected detection times of the proposed policy divided by those of the oracle policy (right), against thresholds $a=b$.
    }
    \label{fig: ratios}
\end{figure}

In Figure \ref{fig: ESS_b'} we set $a=b=20$ and plot the expected detection times of the proposed policy, i.e., $\Exp_B[\hat T_k\wedge \hat T_{\text{stop}}]$ for $k\in[K]$, against $b'$.
Upper lines correspond to greater $k$, and the lines for $5=|B|< k\leq K=10$ basically overlap. 
Note that the best selection of $b'$ slightly increases as $k\leq|B|-1$ increases. 
A reason is that signals with smaller positive drifts are more likely to go below $-b'$ and be missed, so they require larger $b'$ to guarantee that they are sampled until detection in Phase I. 
Overall, the performance of the proposed policy is quite insensitive to the selection of $b'$, and the theoretical selection, $\log(a)$, drawn as a vertical, dashed, gray line in Figure \ref{fig: ESS_b'}, is quite reasonable. 
Thus, we adopt this selection when comparing with other policies in the next numerical study. 

In Figure \ref{fig: comparison} we compare the three policies.
The six subfigures from left to right and from top to bottom correspond to $\Exp_B[ T_k\wedge T_{\text{stop}}]$ for $k=1$,$\ldots$, $k=5$, and $k=6,7,8,9,10$.
The three lines of three different colors correspond to the three policies, as shown in the legend of the first subfigure. 
We can see that the proposed policy outperforms the ``follow-the-leader" policy significantly for $1\leq k\leq |B|-1=4$, and slightly underperforms for $k=|B|=5$.
The latter is reasonable, because exploration is not necessary in minimizing the expected time until detecting all signals.
All policies perform basically the same for $k\geq |B|+1$, which are close to the sum of  $K$ SPRTs. 

In Figure \ref{fig: ratios}, we plot the expected detection times of the proposed policy and the ratios of the expected detection times of the proposed policy divided by those of the oracle policy, against thresholds.
We can see that the former are basically linear in the thresholds, and the latter converge to one, corroborating our asymptotic theory.

\section{Conclusion and future directions} \label{section: conclusion}
In this work, we consider an active signal detection problem, where multiple independent data streams are present,  only one can be observed at every time instant, and it is required to control the two types of familywise error rates below arbitrary, user-specified levels.
The objective that we propose  is to  minimize 
asymptotically, as the target error rates go to $0$,
not only the expected total sample size, but also the  expected time until either the $k_{th}$ signal declaration or termination simultaneously for every $k \in [K]$.
Next, we discuss some potential extensions.

In this work, we postulate two simple hypotheses for every stream, 
and, as a result,  we can order  the streams, as in Section \ref{section: quick detection of signals}, in terms of their detection difficulty. The extension to 
unknown signal and noise distributions is very interesting and important. Such a problem has been studied in 
\cite{Kobi_2020_composite}, where a ``follow-the-leader"  approach (in the spirit of Section \ref{subsec: follow-the-leader}) is studied,  under the assumption that  it is a priori known that there is exactly one signal. With possibly multiple signals, this policy will be efficient in detecting all signals simultaneously, but poor in detecting easier ones earlier, similar to the phenomenon we observed in this work. Other relevant references include 
\cite{OAI, ControlledSensing_Composite}. 
Others directions of interest are the possibility to sample  multiple streams at a time, as well as the availability of prior information regarding the number of signals \cite{Kobi_2015_active, Aris_IEEE}. 
In both cases, it may not be possible to find  a policy that minimizes the same set of objective functions as in this work, and a new problem formulation is needed.
 Finally, another direction is to consider dependent streams and/or non-i.i.d. data within streams. 
\cite{ControlledSensing_NonUniformCost, chaudhuri2024joint} are some related references.

\appendices
\section{Likelihood ratios} \label{appendix, llr}
For any sampling rule $\cS$, two subsets $B,C\subseteq[K]$ and time $n\geq 1$, we denote by $\lambda_{B,C}^\cS(n)$ the log-likelihood ratio between $\Pro_B$ and $\Pro_C$ based on the observations collected by $\cS$ up to time $n$. Based on the i.i.d. assumption within streams and independence assumption across streams, this  can be written as 
\begin{equation*}
\begin{aligned}
    \lambda_{B,C}^\cS(n) := \lambda_{B,C}^\cS(n-1) 
    + \sum_{i\in B\backslash C} \log \left( \frac{g_i(X_i(n))}{f_i(X_i(n))} \right) \,  1\{\cS(n)=i\} 
    - \sum_{i\in C\backslash B} \left( \log\frac{g_i(X_i(n))}{f_i(X_i(n))} \right) \,  1\{\cS(n)=i\},
\end{aligned}
\end{equation*}
with  $\lambda_{B,C}^\cS(0):= 0$. Recalling the statistics 
$\{\lambda_i^\cS(n):i\in[K]\}$ defined in \eqref{LLR_k}, it follows that
\begin{equation*}
    \lambda_{B,C}^\cS(n) = \sum_{i\in B\backslash C} \lambda_i^\cS(n) - \sum_{i\in C\backslash B} \lambda_i^\cS(n),
\end{equation*}
which reveals that 
$\lambda_{B,C}^\cS(n)$ reduces to  $\lambda_i^\cS(n) $  when $B=\{i\}$ and $C=\emptyset$ and reduces to $-\lambda_i^\cS(n)$ when $B=\emptyset$ and $C=\{i\}$.


We also note that for any sampling rule $\cS$ and any  $B,C\subseteq[K]$ the sequence 
\begin{equation*}
    \lambda_{B,C}^\cS(n)
    -\sum_{i\in B\backslash C} I_i N_i^\cS(n) - \sum_{i\in C\backslash B} J_i N_i^\cS(n), \; n \geq 0
\end{equation*}
is an $\cF^\cS$-martingale with mean zero under $\Pro_B$,
where 
$$N_i^\cS(n):=\sum_{t=1}^n 1\{\cS(t)=i\}$$ denotes the number of times stream $i$ is sampled up to time $n$.
As a result, for any $\Pro_B$-integrable $\cF^\cS$-stopping time $T$, by  the optional stopping theorem (see, e.g., \cite[Chapter 13.2]{measure_theory_book}) 
we have 
    \begin{align} \label{walds_id}
    \begin{split}
        & \Exp_B[\lambda_{B,C}^\cS(T)] 
        = \sum_{i\in B\backslash C} I_i \, \Exp_B[N_i^\cS(T)] + \sum_{i\in C\backslash B} J_i \, \Exp_B[N_i^\cS(T)].
    \end{split}
    \end{align}
Besides, by the information-theoretical inequality in \cite[Lemma 3.2.1]{Tartakovsky_2003}, the following holds for any $\cF^\cS(T)$-measurable event $\Gamma$:
\begin{equation} \label{information-theoretical}    \Exp_B[\lambda^\cS_{B,C}(T)] \geq d(\Pro_B(\Gamma),\Pro_C(\Gamma)),
\end{equation}
where $d(x,y)$ has been defined in \eqref{func d} and satisfies $d(x,y)\geq x|\log y|-\log 2$.



\section{Proof of the lower bounds} \label{app: proof of the lower bound}

\begin{IEEEproof}[Proof of Theorem \ref{theorem: ALB}]
Fix $\alpha,\beta\in(0,1)$ so that $\alpha+\beta<1/(K+1)$, $\delta=(S,T_{\text{stop}},D)\in\Delta(\alpha,\beta)$, and $B\subseteq[K]$.
In what follows we assume that $\{T_k\wedge T_{\text{stop}}:k\in[K]\}$ are $\Pro_B$-integrable,
because otherwise the lower bounds hold trivially. 

1) Suppose $B\neq\emptyset$ and fix $1\leq k\leq |B|$.
We start by observing that 
\begin{equation*}
   T_k\wedge T_{\text{stop}} = \sum_{i\in [K]} N_i^S(T_k\wedge T_{\text{stop}}) \geq \sum_{i\in B} N_i^S(T_k\wedge T_{\text{stop}}),
\end{equation*}
where the first equality holds because exactly one source is sampled at each time instant.  For any stream $i\in[K]$ we denote by $\Gamma_i$ the event that stream $i$ is among the first $k$ streams declared as signals, that is,
\begin{equation*}
    \Gamma_i := \{\tau_i\leq T_k\wedge T_{\text{stop}}\},\quad \text{where}
\quad \tau_i:= \inf\{n\geq 1: i\in D(n)\}. 
\end{equation*}
For any $i\in B$, by \eqref{walds_id}-\eqref{information-theoretical} we have
\begin{align} \label{stream-level ineq}
\begin{split}
I_i \, \Exp_B[N^S_i(T_k\wedge T_{\text{stop}})] &\geq d(\Pro_B(\Gamma_i), \Pro_{B\backslash\{i\}}(\Gamma_i)).
\end{split}
\end{align}
Since $\Gamma_i\subseteq\{i\in D(T_{\text{stop}})\}$, which is a type-I error under $\Pro_{B\backslash\{i\}}$, we have $\Pro_{B\backslash\{i\}}(\Gamma_i)\leq\alpha$. 
We claim that 
\begin{equation*}
d\left(\Pro_B(\Gamma_i),\Pro_{B\setminus\{i\}}(\Gamma_i)\right)
\geq
d\left(\Pro_B(\Gamma_i)\vee\alpha,\alpha\right).
\end{equation*}
Indeed, if $\Pro_B(\Gamma_i)\leq\alpha$, then this inequality is trivial since $d(\cdot,\cdot)$ is always non-negative; $\Pro_B(\Gamma_i)>\alpha$, then this inequality holds because the function of $y$, $d(\Pro_B(\Gamma_i),y)$ for $y\leq\alpha$, is decreasing in $y$.
Therefore,
\begin{equation*}
    \Exp_B[T_k\wedge T_{\text{stop}}] \geq \sum_{i\in B} \frac{d(\Pro_B(\Gamma_i)\vee\alpha,\alpha)}{I_i}.
\end{equation*}

Next, we derive a joint condition on $\{\Pro_B(\Gamma_i):i\in B\}$.
Note that on the event $\{D(T_{\text{stop}})=B\}$, all of the first $k$ streams declared as signals belong to $B$, so 
\begin{equation*}
    \sum_{i\in B} 1\{\Gamma_i\} \geq k \,  1\{D(T_{\text{stop}})= B\}.
\end{equation*}
Taking expectations and using the two error constraints, we have
\begin{equation*}
    \sum_{i\in B} \Pro_B(\Gamma_i) \geq k \, \Pro_B(D(T_{\text{stop}}) = B) \geq k\, (1-\Pro_B(D(T_{\text{stop}})\backslash B\neq\emptyset) - \Pro_B(B\backslash D(T_{\text{stop}})\neq\emptyset) ) \geq k(1-\alpha-\beta).
\end{equation*}

Therefore, we obtain 
\begin{equation*}
    \Exp_B[T_k\wedge T_{\text{stop}}] 
    \geq \min_{\substack{\alpha\leq q_i\leq 1,\,i\in B \\ \sum_{i\in B} q_i\geq k(1-\alpha-\beta)}} \sum_{i\in B} \frac{d(q_i,\alpha)}{I_i}
    = \min_{\substack{0\leq p_i\leq 1-\alpha,\,i\in B \\ \sum_{i\in B} p_i\leq |B|-k(1-\alpha-\beta)}} \sum_{i\in B} \frac{d(1-p_i,\alpha)}{I_i}.
\end{equation*}
Since this inequality holds for all $\delta\in\Delta(\alpha,\beta)$, taking the infimum gives \eqref{non-ALB for tauk}.

Next we solve this constrained minimization problem. 
The condition $\alpha+\beta<1/(K+1)$ implies $1-\alpha-\beta>K(\alpha+\beta) > |B|\alpha$, so $k(1-\alpha-\beta)> |B|\alpha$, and hence $|B|-k(1-\alpha-\beta)<|B|(1-\alpha)$, which excludes the trivial minimum $0$ at $p_i=1-\alpha$ for all $i\in B$.
Since the constrained space is convex and, for every $i\in B$, the mapping $p\mapsto d(1-p,\alpha)/I_i$ is strictly convex and strictly decreasing in $[0,1-\alpha]$, the minimum must be achieved uniquely on the boundary of the constrained space. 
The KKT conditions imply that the minimizer is 
\begin{equation*}
\frac{1}{I_i}
\left.
\frac{\partial}{\partial p}d(1-p,\alpha)
\right|_{p=p_i^*}
+\lambda
=
\frac{1}{I_i}
\log
\frac{p_i^*\alpha}
{(1-p_i^*)(1-\alpha)}
+\lambda
=
0 \text{ for all } i\in B,
\end{equation*}
i.e., 
\begin{equation*}
    p_i^* = \frac{1-\alpha}{1-\alpha+\alpha\exp\{\lambda I_i\}} \text{ for all } i\in B,
\end{equation*}
where $\lambda$ is determined by the constraint
\begin{equation*}
    \sum_{i\in B} p_i^* = |B|-k(1-\alpha-\beta).
\end{equation*}

2) We now let $B\neq[K]$ and fix $|B|<k\leq K$.
First, let
\begin{equation*}
    G:=\{T_{|B|}\leq T_k\wedge T_{\text{stop}}\}=\{T_{|B|}\leq T_{\text{stop}}\}.
\end{equation*}
The event $G^c$ implies type-II errors under $\Pro_B$, so $\Pro_B(G) \geq 1-\beta$, whereas, for every $i\in B$, $G$ implies type-I errors under $\Pro_{B\backslash\{i\}}$, so $\Pro_{B\backslash\{i\}}(G)\leq \alpha$.
Thus, by \eqref{walds_id}-\eqref{information-theoretical}
and the fact that $d(x,y)$ increases as $x$ and $y$ becomes farther away provided that one of them is small and the other large, we obtain
\begin{equation*}
    I_i \, \Exp_B[N_i^\cS(T_k\wedge T_{\text{stop}})] \geq d(\Pro_B(G),\Pro_{B\backslash\{i\}}(G)) \geq d(1-\beta,\alpha) \text{ for all } i\in B.
\end{equation*}
Symmetrically, let
\begin{equation*}
    G:=\{T_{|B|+1}> T_k\wedge T_{\text{stop}}\}=\{T_{|B|+1}> T_{\text{stop}}\}.
\end{equation*}
The event $G^c$ implies type-I errors under $\Pro_B$, so $\Pro_B(G)\geq1-\alpha$, whereas, for every $i\notin B$, $G$ implies type-II errors under $\Pro_{B\cup\{i\}}$, so $\Pro_{B\cup\{i\}}(G)\leq\beta$.
Thus,
\begin{equation*}
    J_i \, \Exp_B[N_i^\cS(T_k\wedge T_{\text{stop}})] \geq d(\Pro_B(G), \Pro_{B\cup\{i\}}(G))
    \geq d(1-\alpha,\beta) \text{ for all } i\notin B.
\end{equation*}
Summing these inequalities over all streams yields
\begin{equation*}
\Exp_B[T_k\wedge T_{\text{stop}}] = \sum_{i\in[K]} \Exp_B[N^S_i(T_k\wedge T_{\text{stop}})]
\geq
\sum_{i\in B}\frac{d(1-\beta,\alpha)}{I_i}
+\sum_{i\notin B}\frac{d(1-\alpha,\beta)}{J_i}.
\end{equation*}
Taking the infimum over $\delta\in\Delta(\alpha,\beta)$ proves \eqref{non-ALB for taus}.
\end{IEEEproof}

\begin{IEEEproof} [Proof of Corollary \ref{corollary: ALB}]
Equation \eqref{coro: asym_lower_bound} follows from \eqref{non-ALB for taus} and the fact that $d(1-x,y)\sim |\log y|$ as $x,y\to 0$. 
For equation \eqref{coro: asym_lower_bound_k}, using the inequality $d(1-x,y)\geq (1-x)|\log y|-\log 2$ for all $x,y\in(0,1)$, we have 
\begin{align*}
& \, \cJ_{B,k}(\alpha,\beta)
\geq \min_{\{p_i:i\in B\}\in\Omega_1} \sum_{i\in B} \frac{d(1-p_i, \alpha)}{I_i} \\
\geq & \,  
\min_{\{p_i:i\in B\}\in\Omega_1} \sum_{i\in B} (1-p_i)\, \frac{|\log\alpha|}{I_i}- (\log 2)
\sum_{i\in B} \frac{1}{I_i} \\ 
= & \,  \min_{\{q_i:i\in B\}\in\Omega_2} \sum_{i\in B} q_i\, \frac{|\log\alpha|}{I_i}- (\log 2)
\sum_{i\in B} \frac{1}{I_i},
\end{align*}
where 
\begin{equation*}
    \Omega_1 := \left\{\{p_i:i\in B\}: 0\leq p_i\leq 1-\alpha \text{ for } i\in B \text{ and } \sum_{i\in B} p_i \leq |B|-k(1-\alpha-\beta)\right\},
\end{equation*}
and, after a change of variables $q_i=1-p_i$ for $i\in B$,
\begin{equation*}
    \Omega_2 := \left\{\{q_i:i\in B\}: \alpha\leq q_i\leq 1 \text{ for } i\in B \text{ and } \sum_{i\in B} q_i \geq k(1-\alpha-\beta)\right\}.
\end{equation*}
It is clear that this minimum is attained by assigning weight as much as possible to those $i\in B$ with the largest $I_i$. 
With $\alpha+\beta<1/(K+1)$, this minimum is attained at the $\{q_i:i\in B\}$ with $q_i=1$ for those $i$ corresponding to $I_{(1)}(B), \ldots, I_{(k-1)}(B)$, $q_i=1-k(\alpha+\beta)-(|B|-k)\alpha$ for the $i$ corresponding to $I_{(k)}(B)$, and $q_i=\alpha$ for all other $i$'s. 
Thus, 
\begin{equation*}
\cJ_{B,k}(\alpha,\beta)
\geq |\log\alpha|
\left(
\sum_{i=1}^{k-1}\frac{1}{I_{(i)}(B)}
+\frac{1-k(\alpha+\beta)-(|B|-k)\alpha}{I_{(k)}(B)} 
+\sum_{i=k+1}^{|B|} \frac{\alpha}{I_{(i)}(B)}
\right) -(\log 2)\sum_{i\in B}\frac{1}{I_i}.
\end{equation*}
Letting $\alpha, \beta \to 0$ proves \eqref{coro: asym_lower_bound_k}.

\end{IEEEproof}

\subsection{A loose lower bound}
In what follows, we apply a common approach for proving lower bounds in the literature on active sequential testing (see, e.g., \cite{Aris_IEEE,JMVA2026}) and pure-exploration bandit problems (see, e.g., \cite{kaufmann2016complexity,Overlapping2021}).
The resulting lower bound is valid, but its first-order asymptotic coefficient is generally small when $1<k<|B|$.
The reason is that this approach applies the information-theoretical inequality \eqref{information-theoretical} once to all streams \eqref{global-level},  
whereas the approach we use applies that inequality multiple times to individual streams \eqref{stream-level ineq}.

\begin{proposition} \label{proposition, loose lower bound}
        Fix $\alpha,\beta\in(0,1)$ that $\alpha+\beta<1$ and $B\subseteq[K]$.
    For $1\leq k\leq |B|$ we have
\begin{equation}\label{loose lower bound}
        \cJ_{B,k}(\alpha,\beta)
        \geq  \min_{k\leq r\leq |B|}
        \frac{1}{r-k+1} \,  
\sum_{i=1}^{r}\frac{d(1-\beta,\alpha)}{I_{(i)}(B)}.
    \end{equation}
    Consequently as $\alpha, \beta \to 0$
we have     
    \begin{equation*}
        \cJ_{B,k}(\alpha,\beta)
        \geq 
        \min_{k\leq r\leq |B|} \frac{1}{r-k+1} \sum_{i=1}^r \frac{|\log\alpha|}{I_{(i)}(B)}. 
    \end{equation*}
\end{proposition}
\begin{IEEEproof}
Suppose $B\neq\emptyset$ and fix $1\leq k\leq |B|$.
Let $C \subseteq B$ such that $|C|\leq k-1$. 
By \eqref{walds_id}-\eqref{information-theoretical} we have
\begin{equation} \label{global-level}
     \sum_{i\in B\backslash C} I_i \Exp_B[N_i^S(T_k\wedge T_{\text{stop}})] \geq d(\Pro_B(T_k\leq T_{\text{stop}}), \Pro_C(T_k\leq T_{\text{stop}})) \geq d(1-\beta,\alpha), 
\end{equation}
where the last inequality is because
$\{T_k>T_{\text{stop}}\}$ implies type-II errors when the true number of signals is $|B|\geq k$,
so $\Pro_B(T_k\leq T_{\text{stop}})\geq1-\beta$,
and 
$\{T_k\leq T_{\text{stop}}\}$ implies type-I errors when the true number of signals is $|C|\leq k-1$, 
so $\Pro_C(T_k\leq T_{\text{stop}})\leq\alpha$.
Also note that 
\begin{equation*}
    \sum_{i\in B\backslash C} I_i \, \Exp_B[N_i^S(T_k\wedge T_{\text{stop}})] = \Exp_B[ T_k\wedge T_{\text{stop}}] \sum_{i\in B\backslash C} I_i \, \frac{\Exp_B[N_i^\cS(T_k\wedge T_{\text{stop}})]}{\Exp_B[ T_k\wedge T_{\text{stop}}]},
\end{equation*}
and, since only one stream is sampled at every time instant, 
\begin{equation*} 
    \left\{\frac{\Exp_B[N_i^\cS(T_k\wedge T_{\text{stop}})]}{\Exp_B[T_k\wedge T_{\text{stop}}]}: i\in [K]\right\} \in W_K,
\end{equation*}
where 
 $W_K:=\{w=(w_1,\ldots,w_K)\in[0,1]^K: w_1+\cdots+w_K=1\}
 $ denotes all discrete probability distributions on $[K]$.
    Combining these results and
    taking the worst case over  $C$ we conclude that 
\begin{equation*}
\Exp_B\left[T_k\wedge T_{\text{stop}} \right] \geq
\frac{d(1-\beta,\alpha)}{
\displaystyle
\sup_{w\in W_K}\inf_{\substack{C\subseteq B\\|C|\leq k-1}}
\sum_{i\in B\backslash C}w_iI_i} 
=
\frac{d(1-\beta,\alpha)}{
\displaystyle
\sup_{w\in W_B}\inf_{\substack{C\subseteq B\\|C|=|B|-k+1}}
\sum_{i\in C}w_i I_i},
\end{equation*}
where $W_B:=\{w=(w_i)_{i\in B}\in[0,1]^{|B|}:\sum_{i\in B} w_i=1\}$, and the last equality is because the infimum is attained when the number of elements in the summation is the fewest, 
and the supremum is attained when all weights are within $B$.
The max-min problem in the denominator is evaluated in Lemma \ref{lemma, solution to the max-min problem}, which yields the multiplicative constant in \eqref{loose lower bound}.
\end{IEEEproof}

\section{Proof of the upper bounds} \label{app: proof of the upper bound}

\begin{IEEEproof} [Proof of Theorem \ref{theorem, cohen's}]
  For any $1 \le k \leq |B|$ we have
   \begin{equation*}
    \begin{aligned}
       \hat T_{k}\wedge \hat T_{\text{stop}} 
       = \hat T_{k}\wedge \hat T_{\text{stop}} \cdot 1\{\hat D=B\} + \hat T_{k}\wedge \hat T_{\text{stop}} \cdot 1\{\hat D\neq B\} 
        \leq \hat T_{\text{det}} + \hat T_{\text{stop}} \cdot 1\{\hat D\neq B\},
    \end{aligned}
    \end{equation*}
    where
    \begin{equation*}
    \begin{aligned}
        \hat T_{\text{det}} := \max \{\hat T_k: \hat T_k\leq \hat T_{\text{stop}}, \, k\in[K]\}.
    \end{aligned}
    \end{equation*}
     By H\"older's inequality, for any $r>1$,
    \begin{equation*}
    \begin{aligned}
        \Exp_B[\hat T_{k}\wedge \hat T_{\text{stop}}] 
        \leq & \; \Exp_B[\hat T_{\text{det}}] + \left(\Exp_B[(\hat T_{\text{stop}})^r]\right)^{1/r} \Pro_B(\hat D\neq B)^{1-1/r}.
    \end{aligned}
    \end{equation*}
     From \cite[Section IV.B]{Kobi_2015_active} we know 
    \begin{equation*}
        \Exp_B[\hat T_{\text{det}}] \lesssim \sum_{i\in B} \frac{a}{I_i} \text{ as } a,b\to\infty \text{ so that } b= O(a).
      \end{equation*}   
By Theorem \ref{theorem, upper bound on Tstopr} and \ref{theorem, error control of the general decision rule}, we have 
\begin{equation*}
   \begin{gathered}
        (\Exp_B[( \hat T_{\text{stop}})^r])^{1/r}= O(a\vee b) \text{ as } a,b\to\infty,\\
        \Pro_B(\hat D\neq B)\leq |B|e^{-b} + (K-|B|)e^{-a} \leq K e^{-a\wedge b}.
    \end{gathered}
    \end{equation*}
    The desired result follows after plugging in.
\end{IEEEproof}

\begin{IEEEproof} [Proof of Theorem \ref{theorem: AUB}]
Fix $B\subseteq[K]$, $a>0$, $0\leq b'\leq b$,  $\epsilon\in(0,1)$ and $1\leq k\leq |B|$. 
    We decompose
\begin{equation} \label{decom}
    \begin{aligned}
         \hat T_k\wedge  \hat T_{\text{stop}} =  \hat T_k\wedge  \hat T_{\text{stop}} \cdot 1\{\Gamma\} +  \hat T_k\wedge  \hat T_{\text{stop}} \cdot 1\{\Gamma^c\}
         \leq   \hat T_k \cdot 1\{\Gamma\} +   \hat T_{\text{stop}} \cdot 1\{\Gamma^c\}
    \end{aligned}
    \end{equation}
where $\Gamma$ is the event  that the streams that have been declared as signals in Phase I are those in $B$, i.e.,
\begin{equation*}
        \Gamma := \{\hat D(\tau_K) =B\}.
    \end{equation*}

    We start with the first term. 
On the event $\Gamma$, at the time of the $k$th detection, 
  $\hat{T}_k$, the streams that have been sampled are the   $k$ streams in $B$ with the smallest indexes, $B_k:=\{l\in B: l\leq i_k(B)\}$, where $i_k(B)$ is the  $k_{th}$ smallest index in $B$, and the streams in $B^c$ whose indices are smaller than $i_k(B)$, $B'_k:=\{l\notin B: l\leq i_k(B)\}$. The LLRs of those in $B_k$ have exceeded the upper threshold $a$, and the  LLRs of those in $B'_k$ have exceeded the lower threshold $-b'$. Thus, on the event $\Gamma$, 
    \begin{equation*}
    \begin{aligned}
         \hat T_k = \sum_{i\in B_k}(\inf\{n\geq \tau_{i-1}+1: \lambda_i^{ \cS}(n)\geq a\} - \tau_{i-1}) 
         + \sum_{i\in B'_k}(\inf\{n\geq \tau_{i-1}+1: \lambda_i^{ \cS}(n)\leq -b'\} - \tau_{i-1}),
    \end{aligned}
    \end{equation*}
    which
    has the same distribution under $\Pro_B$ as 
    \begin{equation*}
    \begin{aligned}
        \sum_{i\in B_k}\inf\{n\geq 1: \lambda_i(n)\geq a\} +\sum_{i\in B'_k}\inf\{n\geq 1: \lambda_i(n)\leq -b'\}.
    \end{aligned}
    \end{equation*}
    Based on Lemma \ref{Lemma: for AUB}, this is upper bounded by 
    \begin{equation*}
    \begin{aligned}
        \sum_{i\in B_k} \left(\frac{a}{I_i} \cdot \frac{1}{1-\epsilon} + \cV_i^+(\epsilon)\right) + \sum_{i\in B'_k} \left(\frac{b'}{J_i} \cdot \frac{1}{1-\epsilon} + \cV_i^-(\epsilon)\right) 
        \leq \left(\sum_{i=1}^k \frac{a}{I_{(i)}(B)} + \sum_{i\notin B} \frac{b'}{J_i} \right) \cdot \frac{1}{1-\epsilon} + \cV_B(\epsilon).
    \end{aligned}
    \end{equation*}
    Taking expectations we obtain
    \begin{equation*}
    \begin{aligned}
        \Exp_B[ \hat T_k \cdot 1\{\Gamma\}] 
        \leq \left(\sum_{i=1}^k \frac{a}{I_{(i)}(B)} + \sum_{i\notin B} \frac{b'}{J_i} \right) \cdot \frac{1}{1-\epsilon} + V_B(1,\epsilon).
    \end{aligned}
    \end{equation*}

    We continue with the second term in \eqref{decom}.
    By H\"older's inequality, for any $r>1$,
    \begin{equation*}
        \Exp_B[\hat T_{\text{stop}}\cdot 1\{\Gamma^c\}] \leq \left(\Exp_B[(\hat T_{\text{stop}})^r]\right)^{1/r} \Pro_B(\Gamma^c)^{1-1/r}.
    \end{equation*}
    By \eqref{E[r]1/r} we know that $(\Exp_B[( \hat T_{\text{stop}})^r])^{1/r}= O(a\vee b)$ for all $r\geq 1$.
    It remains to upper bound $\Pro_B(\Gamma^c)$. 
    Note that
$$ \Gamma^c   = \bigcup_{i \in B} \{ i \notin  \hat D(\tau_K)\} \cup  \bigcup_{i \notin B} \{ i \in  \hat D(\tau_K)\}.
$$
By the union bound, we have 
$$\Pro_B(\Gamma^c)\leq \sum_{i \in B} \Pro_B( i \notin  \hat D(\tau_K)) + \sum_{i \notin B} \Pro_B( i \in  \hat D(\tau_K)),$$
and, similarly to Theorem \ref{theorem, error control of the general decision rule},
$$ \Pro_B(\Gamma^c)\leq |B|e^{-b'}+(K-|B|)e^{-a}\leq K e^{-a\wedge b'}.$$ 
Setting $r=2$ we have $\Pro_B(\Gamma^c)^{1-1/r}=O(e^{-\frac{1}{2}(a\wedge b'}))$.
Combining the above results, the upper bound in \eqref{non-asy upper bound on EB[Tk]} follows. 
\end{IEEEproof}

\section{Supporting lemmas}
\begin{lemma} \label{lemma, solution to the max-min problem}
Let $\{I_i:i\in[K]\}$ be $K\geq1$ non-increasingly ordered positive real numbers, i.e., $I_1\geq\cdots\geq I_K>0$.
Then, for any $1\leq k\leq K$,
\begin{equation}\label{correct max-min}
\sup_{w\in W_K}\inf_{\substack{C\subseteq[K]\\|C|=k}}
\sum_{i\in C}w_iI_i
=
\max_{K-k+1\leq r\leq K}
\frac{r-K+k}{\displaystyle\sum_{i=1}^{r}1/I_i}.
\end{equation}
If $r^*$ is any maximizer on the right-hand-side, the supremum is attained by
\begin{equation*}
    w_i^*=\begin{cases}
    \displaystyle\frac{1/I_i}{\sum_{j=1}^{r^*}1/I_j}, & 1\leq i\leq r^*,\\[1.1em]
    0, & r^*<i\leq K.
    \end{cases}
\end{equation*}
\end{lemma}

\begin{IEEEproof}
For $w\in W_K$, set $x_i:=w_iI_i$.
The inner infimum in \eqref{correct max-min} is the sum of the $k$ smallest values among $x_1,\ldots,x_K$.
We first show that it suffices to consider $w$ such that
\begin{equation}\label{ordered x}
    x_1\geq\cdots\geq x_K.
\end{equation}
Indeed, suppose that $x_i<x_j$ for some $i<j$ and define $w'$ by exchanging the two values $x_i,x_j$, i.e.,
\begin{equation*}
    w'_i:=\frac{x_j}{I_i},\qquad
    w'_j:=\frac{x_i}{I_j},
\end{equation*}
with the other coordinates unchanged.
Note that $\{w'_r I_r:r\in[K]\}$ is the same as $\{w_r I_r:r\in[K]\}$, so the objective function is unchanged, while
\begin{equation*}
\begin{aligned}
w'_i+w'_j-w_i-w_j
&=(x_j-x_i)\left(\frac{1}{I_i}-\frac{1}{I_j}\right)\leq0,
\end{aligned}
\end{equation*}
since $I_i\geq I_j$.
Repeating this operation yields \eqref{ordered x} without increasing the total weight.
If some budget remains, it can be added to the first coordinate, which cannot decrease the objective.
Thus, we may assume \eqref{ordered x} and $\sum_r w_r=1$.

Under \eqref{ordered x}, the inner infimum becomes
\begin{equation*}
    \sum_{i=K-k+1}^{K}x_i.
\end{equation*}
Set $x_{K+1}:=0$ and
\begin{equation*}
    y_r:=x_r-x_{r+1}\geq0,\qquad r\in[K].
\end{equation*}
Then
\begin{equation*}
\begin{aligned}
1
=\sum_{i=1}^{K}\frac{x_i}{I_i}
&=\sum_{r=1}^{K}\left(\sum_{i=1}^{r}\frac{1}{I_i}\right)y_r,
\end{aligned}
\end{equation*}
and
\begin{equation*}
    \sum_{i=K-k+1}^{K}x_i
    =\sum_{r=K-k+1}^{K}(r-K+k)y_r.
\end{equation*}
The variables $y_1,\ldots,y_{K-k}$ consume budget without contributing to the objective, so they can be set equal to zero.
Thus, the whole max-min problem becomes
\begin{equation*}
    \max \sum_{r=K-k+1}^K (r-K+k) y_r, \text{ s.t. } \sum_{r=K-k+1}^K\left( \sum_{i=1}^r\frac{1}{I_i} \right) y_r=1. 
\end{equation*}
It is clear that this maximum is equal to the right-hand-side of \eqref{correct max-min}, and 
if $r^*$ is a maximizer, the value is attained by setting all $y_r$ equal to zero except
\begin{equation*}
    y_{r^*}=\left(\sum_{i=1}^{r^*}\frac{1}{I_i}\right)^{-1}.
\end{equation*}
This gives $x_1=\cdots=x_{r^*}=y_{r^*}$ and $x_i=0$ for $i>r^*$, which yields the stated $w^*$.
\end{IEEEproof}

\begin{lemma} \label{Lemma: for AUB}
 Let $\{S_k(n):n\geq 1\}$, $k\in[K]$ be $K\geq 1$ stochastic processes. Let $\ba=(a_1,\ldots,a_K)\in(0,\infty)^K$ and consider the stopping time
 \begin{equation*}
  T(\ba) := \inf\left\{n\geq 1: S_k(n)\geq a_k \text{ for all } k\in[K]\right\}.
 \end{equation*}
 Then, for any $\boldsymbol{\mu}=(\mu_1,\ldots,\mu_K)\in(0,\infty)^K$ and $\epsilon\in(0,1)$,
 \begin{equation*}
 \begin{aligned}
      T(\ba) & \leq \max_{k\in[K]} \left\{\frac{a_k}{\mu_k}\right\} \cdot \frac{1}{1-\epsilon} + \max_{k\in[K]} \cV_k(\mu_k,\epsilon), 
 \end{aligned}
 \end{equation*}
 where 
 \begin{equation*}
     \cV_k(\mu_k,\epsilon) := \sup\left\{ n\geq 1: S_k(n)/n\leq \mu_k(1-\epsilon) \right\}+1.
 \end{equation*}
\end{lemma}
\begin{IEEEproof} 
    Fix $\boldsymbol{\mu}\in(0,\infty)^K$ and $\epsilon\in(0,1)$, and set 
    $$ \Gamma(\ba,\epsilon):=\left\{  T(\ba)>\max_{k\in[K]}\cV_k(\mu_k,\epsilon) \right\}.  $$
We observe that 
    \begin{align*}
        T(\ba) & = T(\ba)\cdot 1\{\Gamma(\ba,\epsilon)\} + T(\ba)\cdot 1\{\Gamma(\ba,\epsilon)^c\} \\
        &\leq T(\ba)\cdot 1\{\Gamma(\ba,\epsilon)\} + \max_{k\in[K]}\cV_k(\mu_k,\epsilon)\cdot 1\{\Gamma(\ba,\epsilon)^c\}.
    \end{align*}
    On the event $\Gamma(\ba,\epsilon)$, there is a $k'\in [K]$ so that  $$
        ( T(\ba)-1) \mu_{k'}\,  (1-\epsilon) < S_{k'}( T(\ba)-1) < a_{k'},
        $$
        and consequently 
   $$
         T(\ba) < \frac{a_{k'}}{\mu_{k'}} \cdot\frac{1}{1-\epsilon} + 1 \leq  \max_{k\in[K]} \left\{\frac{a_k}{\mu_k}\right\} \cdot \frac{1}{1-\epsilon}+1. $$
  Therefore, we conclude that
\begin{equation*}
    \begin{aligned}
        T(\ba) &
        \leq \max_{k\in[K]} \left\{\frac{a_k}{\mu_k}\right\} \cdot \frac{1}{1-\epsilon} + \max_{k\in[K]}\cV_k(\mu_k,\epsilon).
    \end{aligned}
    \end{equation*}
    
\end{IEEEproof}

\begin{lemma} \label{Lemma: properties of LLR} 
    Let $f,g$ be two densities with respect to $\sigma$-finite measure $\nu$ and assume that $I:=\int g\log\frac{g}{f}d\nu\in(0,\infty)$.
    Let $\{X(n):n\geq 1\}$ be a sequence of i.i.d. random variables and denote by $\{\lambda(n):=\sum_{i=1}^n \log\frac{g(X(i))}{f(X(i))}:n\geq 1\}$ the sequence of log-likelihood ratios. 
    Let $\Pro$ and $\Exp$ be the probability measure and expectation when the common density of $\{X(n):n\geq 1\}$ with respect to $\nu$ is $g$.
    Then, for any $x<I$ and $r\geq 1$, we have
    \begin{equation} \label{rth moment}
        \Exp\left[\sup\{n\geq 1: \lambda(n)/n \leq x\}^r\right] < \infty.
    \end{equation}
\end{lemma}
\begin{IEEEproof}
By the Chernoff bound  \cite[Chapter 3.4]{Dembo_Zeitouni_LDPBook} we have
\begin{equation*} 
    \Pro(\lambda(n)/n\leq x) \leq e^{-n\psi(x)} \text{ for } x\leq I,
\end{equation*}
where 
\begin{equation*} 
    \psi(x) := \sup_{\theta\leq 0}\left\{ \theta x-\log \left(\int g^{1+\theta} f^{-\theta} d\nu\right) \right\}
\end{equation*}
is finite, convex and strictly decreasing at least in $[0,I]$ with $\psi(I)=0$ and $\psi(0):= C>0$ known as the Chernoff information between $f$ and $g$. Therefore, we conclude that 
   \begin{equation*}
    \begin{aligned}
        & \; \Exp[\sup\{t\geq 1: \lambda(t)/t \leq x\}^r] 
        \leq r\sum_{n=1}^\infty n^{r-1}\Pro(\sup\{t\geq 1: \lambda(t)/t \leq x\}\geq n) \\
        = & \; r \sum_{n=1}^\infty n^{r-1}\Pro\left( \exists\,m\geq n: \, \lambda(m)/m\leq x \right)  
        \leq r \sum_{n=1}^\infty n^{r-1} \sum_{m=n}^\infty \Pro(\lambda(m)/m\leq x) \\
        = & \; r \sum_{m=1}^\infty \left(\sum_{n=1}^m n^{r-1}\right) \Pro(\lambda(m)/m\leq x) 
 \leq r \sum_{m=1}^\infty \frac{(m+1)^r}{r} e^{-m\psi(x)} < \infty,
    \end{aligned}
    \end{equation*}
    where in the first inequality we used the inequality that $\Exp[Z^r]\leq r \sum_{n=1}^\infty n^{r-1} \Pro(Z\geq n)$ for non-negative, integer-valued random variable $Z$.  
\end{IEEEproof}

\bibliographystyle{IEEEtran}
\bibliography{main}

@INPROCEEDINGS{ISIT2026_CJ_CS,
  author={Chen, Jinyu and Chen, Shanneng and Xing, Yiming},
  booktitle={2026 IEEE International Symposium on Information Theory (ISIT)}, 
  title={Active Anomaly Detection with Identical Normal and Abnormal Distributions}, 
  year={2026},
  volume={},
  number={},
  pages={1-6},
  doi={10.1109/ISIT62367.2026.11653755}}

@INPROCEEDINGS{ISIT2026_GF,
  author={Xing, Yiming and Fellouris, Georgios},
  booktitle={2026 IEEE International Symposium on Information Theory (ISIT)}, 
  title={Active and Asynchronous Signal Detection with Modified Follow-the-Leader Sampling}, 
  year={2026},
  volume={},
  number={},
  pages={1-6},
  doi={10.1109/ISIT62367.2026.11653900}}

@article{Dragalin1996Metrika,
  title={A simple and effective scanning rule for a multi-channel system},
  author={Vladimir P. Dragalin},
  journal={Metrika},
  year={1996},
  volume={43},
  pages={165-182},
  url={https://api.semanticscholar.org/CorpusID:119859576}
}

@unpublished{Xing_config_2026,
  author = {Xing, Yiming},
  title = {Sequential Multiple Testing with Multiple Hypotheses and Prior Information on the Hypothesis Configuration},
  year = {2026},
  note = {arXiv:2606.00839},
  eprint = {2606.00839},
  archiveprefix = {arXiv},
  primaryclass = {math.ST},
  doi = {10.48550/arXiv.2606.00839},
  url = {https://arxiv.org/abs/2606.00839}
}

@unpublished{SS2026+,
  author = {Xing, Yiming},
  title = {Sequential Multiple Testing of Multiple Composite Hypotheses: An Asymptotic Optimality Theory with General Information Functions},
  year = {2026},
  note = {Forthcoming in \emph{Statistica Sinica}},
  doi = {10.5705/ss.202025.0042},
  url = {https://www3.stat.sinica.edu.tw/statistica/fp/SS-2025-0042.html}
}

@ARTICLE{SPL2026,
  author={Xing, Yiming and Chaudhuri, Anamitra and Chen, Yifan},
  journal={IEEE Signal Processing Letters}, 
  title={Signal Detection Under Composite Hypotheses With Identical Distributions for Signals and for Noises}, 
  year={2026},
  volume={33},
  number={},
  pages={3596-3600},
  doi={10.1109/LSP.2026.3724342}}

@InProceedings{TB2016ICML,
  title = 	 {An optimal algorithm for the Thresholding Bandit Problem},
  author = 	 {Locatelli, Andrea and Gutzeit, Maurilio and Carpentier, Alexandra},
  booktitle = 	 {Proceedings of The 33rd International Conference on Machine Learning},
  pages = 	 {1690--1698},
  year = 	 {2016},
  editor = 	 {Balcan, Maria Florina and Weinberger, Kilian Q.},
  volume = 	 {48},
  series = 	 {Proceedings of Machine Learning Research},
  address = 	 {New York, New York, USA},
  month = 	 {20--22 Jun},
  publisher =    {PMLR},
  url = 	 {https://proceedings.mlr.press/v48/locatelli16.html}
}

@InProceedings{TB2022AISTAT,
  title = 	 { Nearly Optimal Algorithms for Level Set Estimation },
  author =       {Mason, Blake and Jain, Lalit and Mukherjee, Subhojyoti and Camilleri, Romain and Jamieson, Kevin and Nowak, Robert},
  booktitle = 	 {Proceedings of The 25th International Conference on Artificial Intelligence and Statistics},
  pages = 	 {7625--7658},
  year = 	 {2022},
  editor = 	 {Camps-Valls, Gustau and Ruiz, Francisco J. R. and Valera, Isabel},
  volume = 	 {151},
  series = 	 {Proceedings of Machine Learning Research},
  month = 	 {28--30 Mar},
  publisher =    {PMLR},
  url = 	 {https://proceedings.mlr.press/v151/mason22a.html}
}

@article{FirstGAI2026JMLR,
  title={An anytime algorithm for good arm identification},
  author={Jourdan, Marc and Delahaye-Duriez, Andr{\'e}e and R{\'e}da, Cl{\'e}mence},
  journal={Journal of Machine Learning Research},
  volume={27},
  number={19},
  pages={1--90},
  year={2026}
}

@InProceedings{GAI2025AISTAT,
  title = 	 {Reward Maximization for Pure Exploration: Minimax Optimal Good Arm Identification for Nonparametric Multi-Armed Bandits},
  author =       {Cho, Brian M and Meier, Dominik and Gan, Kyra and Kallus, Nathan},
  booktitle = 	 {Proceedings of The 28th International Conference on Artificial Intelligence and Statistics},
  pages = 	 {2251--2259},
  year = 	 {2025},
  editor = 	 {Li, Yingzhen and Mandt, Stephan and Agrawal, Shipra and Khan, Emtiyaz},
  volume = 	 {258},
  series = 	 {Proceedings of Machine Learning Research},
  month = 	 {03--05 May},
  publisher =    {PMLR},
  url = 	 {https://proceedings.mlr.press/v258/cho25a.html}
}

@article{GAI2019ML,
  title={Good arm identification via bandit feedback},
  author={Kano, Hideaki and Honda, Junya and Sakamaki, Kentaro and Matsuura, Kentaro and Nakamura, Atsuyoshi and Sugiyama, Masashi},
  journal={Machine Learning},
  volume={108},
  number={5},
  pages={721--745},
  year={2019},
  publisher={Springer}
}

@article{Overlapping2021,
author = {Aurélien Garivier and Emilie Kaufmann},
title = {Nonasymptotic sequential tests for overlapping hypotheses applied to near-optimal arm identification in bandit models},
journal = {Sequential Analysis},
volume = {40},
number = {1},
pages = {61--96},
year = {2021},
publisher = {Taylor \& Francis},
doi = {10.1080/07474946.2021.1847965},
URL = {https://doi.org/10.1080/07474946.2021.1847965
}
}

@article{JMVA2026,
title = {To minimize the expected total sampling cost in sequential testing about a random vector},
journal = {Journal of Multivariate Analysis},
volume = {215},
pages = {105640},
year = {2026},
issn = {0047-259X},
doi = {https://doi.org/10.1016/j.jmva.2026.105640},
url = {https://www.sciencedirect.com/science/article/pii/S0047259X26000461},
author = {Yiming Xing}
}

@article{Qunzhi2023,
author = {Qunzhi Xu and Yajun Mei},
title = {Asymptotic optimality theory for active quickest detection with unknown postchange parameters},
journal = {Sequential Analysis},
volume = {42},
number = {2},
pages = {150--181},
year = {2023},
publisher = {Taylor \& Francis},
doi = {10.1080/07474946.2023.2187417}
}

@ARTICLE{Aris_TIT2025,
  author={Tsopelakos, Aristomenis and Fellouris, Georgios},
  journal={IEEE Transactions on Information Theory}, 
  title={Sequential Anomaly Identification Under Sampling Constraints for Generalized Error Metrics}, 
  year={2025},
  volume={71},
  number={12},
  pages={9753-9783},
  doi={10.1109/TIT.2025.3622405}}

@ARTICLE{Malloy_Nowak_2013,
  author={Malloy, Matthew L. and Tang, Gongguo and Nowak, Robert D.},
  journal={IEEE Transactions on Information Theory}, 
  title={The Sample Complexity of Search Over Multiple Populations}, 
  year={2013},
  volume={59},
  number={8},
  pages={5039-5050},
  doi={10.1109/TIT.2013.2258071}}

@INPROCEEDINGS{ITW2024,
  author={Xing, Yiming and Yan, Shen and Wang, Ziming},
  booktitle={2024 IEEE Information Theory Workshop (ITW)}, 
  title={High-Dimensional Sequential Testing of Multiple Hypotheses}, 
  year={2024},
  volume={},
  number={},
  pages={384-389},
  doi={10.1109/ITW61385.2024.10806932}}

@inproceedings{Fellouris2013unstructured,
  title={Unstructured sequential testing in sensor networks},
  author={Fellouris, Georgios and Tartakovsky, Alexander},
  booktitle={52nd IEEE Conference on Decision and Control},
  pages={4784--4789},
  year={2013},
  organization={IEEE}
}

@article{Lai2011quickest,
  title={Quickest search over multiple sequences},
  author={Lai, Lifeng and Poor, H Vincent and Xin, Yan and Georgiadis, Georgios},
  journal={IEEE Transactions on Information Theory},
  volume={57},
  number={8},
  pages={5375--5386},
  year={2011},
  publisher={IEEE}
}

@ARTICLE{Ali2016quickest,
  author={Heydari, Javad and Tajer, Ali and Vincent Poor, H.},
  journal={IEEE Transactions on Information Theory}, 
  title={Quickest Linear Search over Correlated Sequences}, 
  year={2016},
  volume={62},
  number={10},
  pages={5786-5808},
  doi={10.1109/TIT.2016.2593772}}

@article{Ref4PatternMining_2017,
  title={A survey of sequential pattern mining},
  author={Fournier-Viger, Philippe and Lin, Jerry Chun-Wei and Kiran, Rage Uday and Koh, Yun Sing and Thomas, Rincy},
  journal={Data Science and Pattern Recognition},
  volume={1},
  number={1},
  pages={54--77},
  year={2017}
}

@article{Ref4GeneAsso_2021,
  title={Genome-wide association studies},
  author={Uffelmann, Emil and Huang, Qin Qin and Munung, Nchangwi Syntia and De Vries, Jantina and Okada, Yukinori and Martin, Alicia R and Martin, Hilary C and Lappalainen, Tuuli and Posthuma, Danielle},
  journal={Nature Reviews Methods Primers},
  volume={1},
  number={1},
  pages={59},
  year={2021},
  publisher={Nature Publishing Group UK London}
}

@article{Ref4FraudDet_2022,
  title={Financial fraud: a review of anomaly detection techniques and recent advances},
  author={Hilal, Waleed and Gadsden, S Andrew and Yawney, John},
  journal={Expert systems With applications},
  volume={193},
  pages={116429},
  year={2022},
  publisher={Elsevier}
}

@ARTICLE{Ref4SpectrumSensing_2016,
  author={Geng, Jun and Xu, Weiyu and Lai, Lifeng},
  journal={IEEE Transactions on Signal Processing}, 
  title={Quickest Sequential Multiband Spectrum Sensing With Mixed Observations}, 
  year={2016},
  volume={64},
  number={22},
  pages={5861-5874},
  doi={10.1109/TSP.2016.2602802}}

@ARTICLE{Ana_2024TIT,
  author={Chaudhuri, Anamitra and Fellouris, Georgios and Tajer, Ali},
  journal={IEEE Transactions on Information Theory}, 
  title={Round Robin Active Sequential Change Detection for Dependent Multi-Channel Data}, 
  year={2024},
  volume={70},
  number={12},
  pages={9327-9351},
  doi={10.1109/TIT.2024.3475394}}

@ARTICLE{Qunzhi_2021TIT,
  author={Xu, Qunzhi and Mei, Yajun and Moustakides, George V.},
  journal={IEEE Transactions on Information Theory}, 
  title={Optimum Multi-Stream Sequential Change-Point Detection With Sampling Control}, 
  year={2021},
  volume={67},
  number={11},
  pages={7627-7636},
  doi={10.1109/TIT.2021.3074961}}

@article{Fellouris_2024_CDwithCS,
title = "Quickest Change Detection with Controlled Sensing",
author = "Veeravalli, {Venugopal V.} and Georgios Fellouris and Moustakides, {George V.}",
year = "2024",
doi = "10.1109/JSAIT.2024.3362324",
language = "English (US)",
volume = "5",
pages = "1--11",
journal = "IEEE Journal on Selected Areas in Information Theory",
issn = "2641-8770",
publisher = "Institute of Electrical and Electronics Engineers Inc.",
}

@ARTICLE{Kobi_2022_switching,
  author={Lambez, Tidhar and Cohen, Kobi},
  journal={IEEE Transactions on Signal Processing}, 
  title={Anomaly Search With Multiple Plays Under Delay and Switching Costs}, 
  year={2022},
  volume={70},
  number={},
  pages={174-189},
  doi={10.1109/TSP.2021.3136810}}

@ARTICLE{Kobi_2014_index,
  author={Cohen, Kobi and Zhao, Qing and Swami, Ananthram},
  journal={IEEE Transactions on Signal Processing}, 
  title={Optimal Index Policies for Anomaly Localization in Resource-Constrained Cyber Systems}, 
  year={2014},
  volume={62},
  number={16},
  pages={4224-4236},
  doi={10.1109/TSP.2014.2332982}}

@ARTICLE{Kobi_2019_nonlinear,
  author={Gurevich, Andrey and Cohen, Kobi and Zhao, Qing},
  journal={IEEE Transactions on Signal Processing}, 
  title={Sequential Anomaly Detection Under a Nonlinear System Cost}, 
  year={2019},
  volume={67},
  number={14},
  pages={3689-3703},
  doi={10.1109/TSP.2019.2918981}}

@inproceedings{Murphy_sampling,
 author = {Kaufmann, Emilie and Koolen, Wouter M and Garivier, Aur\'{e}lien},
 booktitle = {Advances in Neural Information Processing Systems},
 editor = {S. Bengio and H. Wallach and H. Larochelle and K. Grauman and N. Cesa-Bianchi and R. Garnett},
 pages = {},
 publisher = {Curran Associates, Inc.},
 title = {Sequential Test for the Lowest Mean: From Thompson to Murphy Sampling},
 url = {https://proceedings.neurips.cc/paper_files/paper/2018/file/7c78335a8924215ea5c22fda1aac7b75-Paper.pdf},
 volume = {31},
 year = {2018}
}

@article{kaufmann2016complexity,
  title={On the complexity of best-arm identification in multi-armed bandit models},
  author={Kaufmann, Emilie and Capp{\'e}, Olivier and Garivier, Aur{\'e}lien},
  journal={The Journal of Machine Learning Research},
  volume={17},
  number={1},
  pages={1--42},
  year={2016},
  publisher={JMLR. org}
}

@article{Kobi_2023_hiera,
  title={Anomaly search over discrete composite hypotheses in hierarchical statistical models},
  author={Gafni, Tomer and Wolff, Benjamin and Revach, Guy and Shlezinger, Nir and Cohen, Kobi},
  journal={IEEE Transactions on Signal Processing},
  volume={71},
  pages={202--217},
  year={2023},
  publisher={IEEE}
}

@inproceedings{OAI,
  title={Optimal best arm identification with fixed confidence},
  author={Garivier, Aur{\'e}lien and Kaufmann, Emilie},
  booktitle={Conference on Learning Theory},
  pages={998--1027},
  year={2016},
  organization={PMLR}
}

@article{Chernoff1959,
 ISSN = {00034851, 21688990},
 URL = {http://www.jstor.org/stable/2237415},
 author = {Herman Chernoff},
 journal = {The Annals of Mathematical Statistics},
 number = {3},
 pages = {755--770},
 publisher = {Institute of Mathematical Statistics},
 title = {Sequential Design of Experiments},
 urldate = {2024-11-25},
 volume = {30},
 year = {1959}
}

@article{ControlledSensing,
  title={Controlled sensing for multihypothesis testing},
  author={Nitinawarat, Sirin and Atia, George K and Veeravalli, Venugopal V},
  journal={IEEE Transactions on automatic control},
  volume={58},
  number={10},
  pages={2451--2464},
  year={2013},
  publisher={IEEE}
}

@article{ControlledSensing_NonUniformCost,
  title={Controlled sensing for sequential multihypothesis testing with controlled Markovian observations and non-uniform control cost},
  author={Nitinawarat, Sirin and Veeravalli, Venupogal V},
  journal={Sequential Analysis},
  volume={34},
  number={1},
  pages={1--24},
  year={2015},
  publisher={Taylor \& Francis}
}

@article{ControlledSensing_Composite,
  title={Sequential controlled sensing for composite multihypothesis testing},
  author={Deshmukh, Aditya and Veeravalli, Venugopal V and Bhashyam, Srikrishna},
  journal={Sequential Analysis},
  volume={40},
  number={2},
  pages={259--289},
  year={2021},
  publisher={Taylor \& Francis}
}

@article{Naghshvar2013,
author = {Mohammad Naghshvar and Tara Javidi},
title = {{Active sequential hypothesis testing}},
volume = {41},
journal = {The Annals of Statistics},
number = {6},
publisher = {Institute of Mathematical Statistics},
pages = {2703 -- 2738},
year = {2013},
doi = {10.1214/13-AOS1144},
URL = {https://doi.org/10.1214/13-AOS1144}
}

@book{measure_theory_book,
  title={Measure theory and probability theory},
  author={Athreya, Krishna B and Lahiri, Soumendra N},
  volume={19},
  year={2006},
  publisher={Springer}
}

@article{PaperIII,
  title={Asymptotically optimal sequential multiple testing with asynchronous decisions},
  author={Xing, Yiming and Fellouris, Georgios},
  journal={Bernoulli},
  volume={31},
  number={1},
  pages={271--294},
  year={2025},
  publisher={Bernoulli Society for Mathematical Statistics and Probability}
}

@ARTICLE{PaperII,
  author={Xing, Yiming and Fellouris, Georgios},
  journal={IEEE Transactions on Information Theory}, 
  title={Signal Recovery With Multistage Tests and Without Sparsity Constraints}, 
  year={2023},
  volume={69},
  number={11},
  pages={7220-7245},
  doi={10.1109/TIT.2023.3299874}}

@article{chaudhuri2024joint,
  title={Joint sequential detection and isolation for dependent data streams},
  author={Chaudhuri, Anamitra and Fellouris, Georgios},
  journal={The Annals of Statistics},
  volume={52},
  number={5},
  pages={1899--1926},
  year={2024},
  publisher={Institute of Mathematical Statistics}
}

@ARTICLE{Aris_IEEE,
  author={Tsopelakos, Aristomenis and Fellouris, Georgios},
  journal={IEEE Transactions on Information Theory}, 
  title={Sequential Anomaly Detection Under Sampling Constraints}, 
  year={2023},
  volume={69},
  number={12},
  pages={8126-8146},
  doi={10.1109/TIT.2022.3177142}}

@article{Bartroff_2010,
author = { Jay   Bartroff  and  Tze   Leung   Lai },
title = {Multistage Tests of Multiple Hypotheses},
journal = {Communications in Statistics - Theory and Methods},
volume = {39},
number = {8-9},
pages = {1597-1607},
year  = {2010},
publisher = {Taylor & Francis},
doi = {10.1080/03610920802592852},
URL = { 
        https://doi.org/10.1080/03610920802592852},
eprint = { 
        https://doi.org/10.1080/03610920802592852}
}

@article{Bartroff_2014_FWER,
  title={Sequential tests of multiple hypotheses controlling type I and II familywise error rates},
  author={Bartroff, Jay and Song, Jinlin},
  journal={Journal of Statistical Planning and Inference},
  volume={153},
  pages={100--114},
  year={2014},
  publisher={Elsevier}
}

@article{De_Baron_Seq_Bonf,
author = { Shyamal K.   De  and  Michael   Baron },
title = {Sequential Bonferroni Methods for Multiple Hypothesis Testing with Strong Control of Family-Wise Error Rates I and II},
journal = {Sequential Analysis},
volume = {31},
number = {2},
pages = {238-262},
year  = {2012},
publisher = {Taylor & Francis},
doi = {10.1080/07474946.2012.665730},
URL = { 
        https://doi.org/10.1080/07474946.2012.665730
},
eprint = { 
        https://doi.org/10.1080/07474946.2012.665730
}
}

@article{Step_up_down,
title = {Step-up and step-down methods for testing multiple hypotheses in sequential experiments},
journal = {Journal of Statistical Planning and Inference},
volume = {142},
number = {7},
pages = {2059-2070},
year = {2012},
issn = {0378-3758},
doi = {https://doi.org/10.1016/j.jspi.2012.02.005},
url = {https://www.sciencedirect.com/science/article/pii/S037837581200050X},
author = {Shyamal K. De and Michael Baron}
}

@ARTICLE{Fellouris_noniid,
  author={Fellouris, Georgios and Tartakovsky, Alexander G.},
  journal={IEEE Transactions on Information Theory}, 
  title={Multichannel Sequential Detection—Part I: Non-i.i.d. Data}, 
  year={2017},
  volume={63},
  number={7},
  pages={4551-4571},
  doi={10.1109/TIT.2017.2689785}}

@article{Tartakovsky_2003,
  title={Sequential detection of targets in multichannel systems},
  author={Tartakovsky, Alexander G and Li, X Rong and Yaralov, George},
  journal={IEEE Transactions on Information Theory},
  volume={49},
  number={2},
  pages={425--445},
  year={2003},
  publisher={IEEE}
}

@article{Kobi_2015_AO,
  title={Asymptotically optimal anomaly detection via sequential testing},
  author={Cohen, Kobi and Zhao, Qing},
  journal={IEEE Transactions on Signal Processing},
  volume={63},
  number={11},
  pages={2929--2941},
  year={2015},
  publisher={IEEE}
}

@article{Kobi_2015_active,
  title={Active hypothesis testing for anomaly detection},
  author={Cohen, Kobi and Zhao, Qing},
  journal={IEEE Transactions on Information Theory},
  volume={61},
  number={3},
  pages={1432--1450},
  year={2015},
  publisher={IEEE}
}

@article{Kobi_2018_heterogeneous,
  title={Active anomaly detection in heterogeneous processes},
  author={Huang, Boshuang and Cohen, Kobi and Zhao, Qing},
  journal={IEEE Transactions on Information Theory},
  volume={65},
  number={4},
  pages={2284--2301},
  year={2018},
  publisher={IEEE}
}

@article{Kobi_2020_composite,
  title={Searching for anomalies over composite hypotheses},
  author={Hemo, Bar and Gafni, Tomer and Cohen, Kobi and Zhao, Qing},
  journal={IEEE Transactions on Signal Processing},
  volume={68},
  pages={1181--1196},
  year={2020},
  publisher={IEEE}
}

@article{Ref4MissleDet_2004,
  title={Boost-phase identification of theater ballistic missiles using radar measurements},
  author={Almogi-Nadler, Meirav and Oshman, Yaakov and Ben-Asher, Joseph Z},
  journal={Journal of Guidance, Control, and Dynamics},
  volume={27},
  number={2},
  pages={197--208},
  year={2004}
}

@article{Ref4AnomalyDet_2009,
  title={Anomaly detection: A survey},
  author={Chandola, Varun and Banerjee, Arindam and Kumar, Vipin},
  journal={ACM Computing Surveys (CSUR)},
  volume={41},
  number={3},
  pages={1--58},
  year={2009},
  publisher={ACM New York, NY, USA}
}

@book{Tartakovsky_Book,
author = {Tartakovsky, Alexander and Nikiforov, Igor and Basseville, Michele},
title = {Sequential Analysis: Hypothesis Testing and Changepoint Detection},
year = {2014},
isbn = {1439838208},
publisher = {Chapman \& Hall/CRC},
edition = {1st}
}

@book{Wald_Book,
author={Abraham Wald},
title={Sequential Analysis},
year={1947},
publisher={John Wiley \& Sons},
address={New York}
}

@article{Malloy_Nowak_2014,
  title={Sequential testing for sparse recovery},
  author={Malloy, Matthew L and Nowak, Robert D},
  journal={IEEE Transactions on Information Theory},
  volume={60},
  number={12},
  pages={7862--7873},
  year={2014},
  publisher={IEEE}
}

@book{Dembo_Zeitouni_LDPBook,
  title={Large Deviations Techniques and Applications},
  author={Amir Dembo and Ofer Zeitouni},
  journal={Applications of Mathematics},
  publisher={Springer, Berlin, Heidelberg},
  year={1998}
}

@book{Bartroff_Book_Clinicaltrials,
  title={Sequential experimentation in clinical trials: design and analysis},
  author={Bartroff, Jay and Lai, Tze Leung and Shih, Mei-Chiung},
  volume={298},
  year={2012},
  publisher={Springer Science \& Business Media}
}

@article{Song_prior,
author = {Yanglei Song and Georgios Fellouris},
title = {{Asymptotically optimal, sequential, multiple testing procedures with prior information on the number of signals}},
volume = {11},
journal = {Electronic Journal of Statistics},
number = {1},
publisher = {Institute of Mathematical Statistics and Bernoulli Society},
pages = {338 -- 363},
year = {2017},
doi = {10.1214/17-EJS1223},
URL = {https://doi.org/10.1214/17-EJS1223}
}

@article{Song_AoS,
author = {Yanglei Song and Georgios Fellouris},
title = {{Sequential multiple testing with generalized error control: An asymptotic optimality theory}},
volume = {47},
journal = {The Annals of Statistics},
number = {3},
publisher = {Institute of Mathematical Statistics},
pages = {1776 -- 1803},
year = {2019},
doi = {10.1214/18-AOS1737},
URL = {https://doi.org/10.1214/18-AOS1737}
}

\end{document}